%% file: root_sttt.tex
\documentclass[sttt,vecarrow]{svjour}
\usepackage{fainekos-macros}
\usepackage{amssymb,amsfonts,wrapfig,floatflt}
\usepackage{amsmath}
\usepackage{graphicx}
\usepackage{algorithm}
\usepackage{algpseudocode}
\graphicspath{{figures/}}
\usepackage{multirow}
\RequirePackage{todonotes}
\usepackage{enumerate}
\usepackage{ifthen}
\usepackage{xcolor,colortbl}
\usepackage{float}

\definecolor{Gray}{gray}{0.9}

\definecolor{revComments}{rgb}{0,0,0}

\newboolean{TECHREP}
\setboolean{TECHREP}{true}

\usepackage{graphics}

\begin{document}
\title{Mining Parametric Temporal Logic Properties \\ in Model Based Design for Cyber-Physical Systems}
\ifthenelse{\boolean{TECHREP}}{\subtitle{Extended Version}}{}

\author{Bardh~Hoxha\and Adel~Dokhanchi\and Georgios~Fainekos
%
}                     
%
%
\institute{Arizona State University \\ \email{\{bhoxha,adokhanc,fainekos\}@asu.edu}}
%
\date{Received: date / Revised version: date}
%
\titlerunning{Mining Parametric Temporal Logic Properties in MBD for CPS}
\authorrunning{Hoxha et al.:} 

\maketitle
\begin{abstract}
One of the advantages of adopting a Model Based Development (MBD) process is that it enables testing and verification at early stages of development. However, it is often desirable to not only verify/falsify certain formal system specifications, but also to automatically explore the properties that the system satisfies. In this work, we present a framework that enables property exploration for Cyber-Physical Systems.
Namely, given a parametric specification with multiple parameters, our solution can automatically infer the ranges of parameters for which the property does not hold on the system. In this paper, we consider parametric specifications in Metric or Signal Temporal Logic (MTL or STL). Using robust semantics for MTL, the parameter mining problem can be converted into a Pareto optimization problem for which we can provide an approximate solution by utilizing stochastic optimization methods. We include algorithms for the exploration and visualization of multi-parametric specifications. The framework is demonstrated on an industrial size, high-fidelity engine model as well as examples from related literature.
\end{abstract}
%
\input{intro}

\input{prob}

\input{prelim}

\input{parameter}

\input{parameterDomain}

\input{examples}

\input{related}

\input{conclusions}

\begin{acknowledgement}
This work has been partially supported by award NSF CNS 1116136 and CNS 1350420. 
Also, we thank the \textit{Toyota Technical Center} for donating a license for the Simuquest Enginuity tool package.
\end{acknowledgement}

\bibliographystyle{splncs}
\bibliography{bibref_database}

\end{document}

%% file: intro.tex

\section{Introduction}
Testing, verification and validation of Cyber-Physical Systems (CPS) is a challenging problem. Prime examples of such systems are aircraft, cars and medical devices which are also safety-critical systems. 
The complexity in these systems arises mostly from the complex interactions between the numerous components (e.g. software enabled controllers) and the physical environment (plant). 
Many accidents \cite{ariane5,Hoffman99near} and recalls in the industry have reinforced the need for better methodologies in this area. 
In addition, general trends indicate that software complexity in CPS is going to increase in the future \cite{montalk1991computer}. 

A recent shift in system development, aimed to alleviate some of the challenges, is the Model Based Design (MBD) paradigm. 
One of the benefits of MBD is that a significant amount of testing and verification of the system can be conducted in various stages of model development. 
{\color{revComments} This is different from the traditional approach, where most of the analysis is conducted on a prototype of the system.}
Due to the importance of the problem, there has been a substantial level of research on testing and verification of models of embedded and hybrid systems (see \cite{TripakisD09model,kapinski2015simulation} for an overview).

In \cite{NghiemSFIGP10hscc,AbbasFSIG11tecs}, the authors propose an approach to support the testing and verification process in MBD. 
The papers provide a new method for testing embedded and hybrid systems against formal requirements which are defined in Metric Temporal Logic (MTL) \cite{Koymans90}. 
{\color{revComments} MTL formulas are interpreted over trajectories/behaviors of the system. In this context, MTL specifications are equivalent to Signal Temporal Logic (STL) \cite{MalerNickovic04} specifications.}
Given a system and an MTL specification, the method searches for operating conditions such that the MTL specification is not satisfied or, in other words, falsified. 
\ifthenelse{\boolean{TECHREP}}{
The authors utilize the concept of system robustness of MTL specifications \cite{FainekosP06fates,FainekosP09tcs} to turn the falsification problem into an optimization problem. 
The notion of the robustness metric enables system developers to measure by how far a system behavior is from failing to satisfy a requirement. 
This allows for the development of an automatic test case generator, which uses a stochastic optimization engine to find operating conditions that falsify the system in terms of the MTL specifications.
\ifthenelse{\boolean{TECHREP}}{The resulting optimization problem may be both non-linear and non-convex.}{}
 To solve the problem, in \cite{SankaranarayananF2012hscc,AnnapureddyF10iecon,NghiemSFIGP10hscc}, the authors present stochastic optimization techniques that solve the falsification problem with very good performance in both accuracy and number of simulations required. 
}
{
The authors utilize the concept of system robustness of MTL specifications \cite{FainekosP06fates,FainekosP09tcs} to turn the falsification problem into an optimization problem. The optimization problem may be both non-linear and non-convex. To solve the problem, several works \cite{SankaranarayananF2012hscc,AnnapureddyF10iecon,NghiemSFIGP10hscc} have presented stochastic techniques that solve the problem with high accuracy. 
}

In \cite{YangHF12ictss}, the authors utilize this notion of robustness to explore and determine system properties. In more detail, given a parameterized MTL specification \cite{AsarinDMN12rv}, where there is an unknown state and/or timing parameter, the authors find the range of values for the parameter such that the system is not satisfied. 

In this work, we extend and generalize the work in \cite{YangHF12ictss} to enable multiple parameter mining and analysis of parametric MTL specifications. 
We improve the efficiency of the previous algorithm in \cite{YangHF12ictss} and present a parameter mining framework for MBD. 
Such an exploration framework would be of great value to the practitioner. 
The benefits are twofold. 
One, it allows for the analysis and development of specifications. 
In many cases, system requirements are not well formalized by the initial system design stages. 
Two, it allows for the analysis and exploration of system behavior. 
If a specification can be falsified, then it is natural to inquire for the range of parameter values that cause falsification.
That is, in many cases, the system design may not be modified, but the guarantees provided should be updated. 

The extension to multiple parameter mining of MTL specifications allows practitioners to use this method with more complex specifications. 
\ifthenelse{\boolean{TECHREP}}{However, as the number of parameters in the specification increases, so does the complexity of the resulting optimization problem.}{}
In the case of single parameter mining, the solution of the problem is a one dimensional range. 
On the other hand, with multiple parameters, finding a solution to the problem becomes more challenging since the optimization problem is converted to a multi-objective optimization problem where the goal is to determine the Pareto front \cite{myers2016response}. 
To solve this problem, we present a method for effective one-sided exploration of the Pareto front and provide a visualization method for the analysis of parameters. 
The algorithms presented in this work are incorporated in the testing and verification toolbox \staliro \cite{AnnapureddyLFS11tacas,staliro:Online}. 
For an overview of the toolbox see \cite{hoxha2014towards}.  
Finally, we demonstrate our framework on a challenge problem from the industry on an industrial scale model and present experimental results on several benchmark problems.
\ifthenelse{\boolean{TECHREP}}{
{\color{revComments} Even though our examples and case study are from the automotive domain, our results can be applied to any application domain where Model Based Design (MBD) and temporal logic requirements are utilized, e.g., medical devices \cite{sankaranarayanan2011model,SankaranarayananF2012cmsb,JiangPM12ieee,ChenDKM13hscc}. }
}{
	{\color{revComments} Our results may be applied to any application domain where temporal logics are utilized, e.g., medical devices \cite{SankaranarayananF2012cmsb,JiangPM12ieee,ChenDKM13hscc}. }

	\vspace{-10pt}
}

\textbf{Summary of Contributions: }
\begin{itemize}
\item We extend and generalize the parameter mining problem presented in \cite{YangHF12ictss}.
\item We provide an efficient solution to the problem of multiple parameter mining.
\item We present two algorithms to explore the Pareto front of parametric MTL with multiple parameters.
\item We illustrate our method with an industrial size case study of a high-fidelity engine model.
\item The algorithms presented in this work are publicly available through our toolbox \staliro \cite{staliro:Online}.
\end{itemize}
\vspace{-5pt}

%% file: prob.tex
\section{Problem Formulation}

\subsection{Preliminaries}
In the rest of the paper, we take a general approach to modeling real-time embedded systems that interact with physical systems that have non-trivial dynamics. 
\ifthenelse{\boolean{TECHREP}}{A major source of complexity in the analysis of these systems arises from the interaction between the embedded system and the physical world.}{}


We fix $N \subseteq \Ne$, where $\Ne$ is the set of natural numbers, to be a finite set of indexes for the finite representation of a system behavior.
In the following, given two sets $A$ and $B$, $B^A$ denotes the set of all functions from $A$ to $B$.
That is, for any $f \in B^A$ we have $f : A \rightarrow B$.
We consider a system $\Sys$ as a mapping from a compact set of {\it initial operating conditions} $\genSS_0$ and {\it input signals} $\genInpSet \subseteq \Inp^N$ to {\it output signals} $\genOut^N$ and {\it timing} (or {\it sampling}) functions $\Sam \subseteq \preals^N$. 
Here, $\Inp$ is a compact set of possible input values at each point in time (input space), $\genOut$ is the set of output values (output space), $\Re$ is the set of real numbers and $\preals$ the set of positive reals. 

We impose three assumptions/restrictions on the systems that we consider:

\begin{enumerate}
\item The input signals (if any) must be parameterizable using a finite number of parameters.
That is, there exists a function $\Ur$ such that for any $u \in \genInpSet$, there exist two parameter vectors $\vec{\lambda} = [\lambda_1$ $\ldots$ $\lambda_m]^\intercal \in \Lambda$, where $\Lambda$ is a compact set, and $t = [t_1$ $\ldots$ $t_m]^\intercal \in \preals^m$ such that {\color{revComments}$m$ is typically much smaller than the maximum number of indices in $N$} and for all $i \in N$, $u(i) = \Ur(\lambda,t)(i)$. 
\label{ass:1}

\item The output space $\genOut$ must be equipped with a generalized metric $\genMet$ which contains a subspace $\genOutTwo$ equipped with a metric $d$ \cite{AbbasFSIG11tecs}.
\label{ass:2}

\item  
For a specific initial condition $\gensspt_0$ and input signal $\InpSig$, there must exist a unique output signal $\gentraj$ defined over the time domain $R$. 
That is, the system $\Sys$ is deterministic.
\label{ass:3}


\end{enumerate}
Further details on the necessity and implications of the aforementioned assumptions can be found in \cite{AbbasFSIG11tecs}. Assumption 3 can also be relaxed as shown in \cite{abbas2014robustness}.


\ifthenelse{\boolean{TECHREP}}{
\begin{figure*}
\centering
\vspace{-10pt}
\begin{tabular}{cc}
\includegraphics[width=7.5cm]{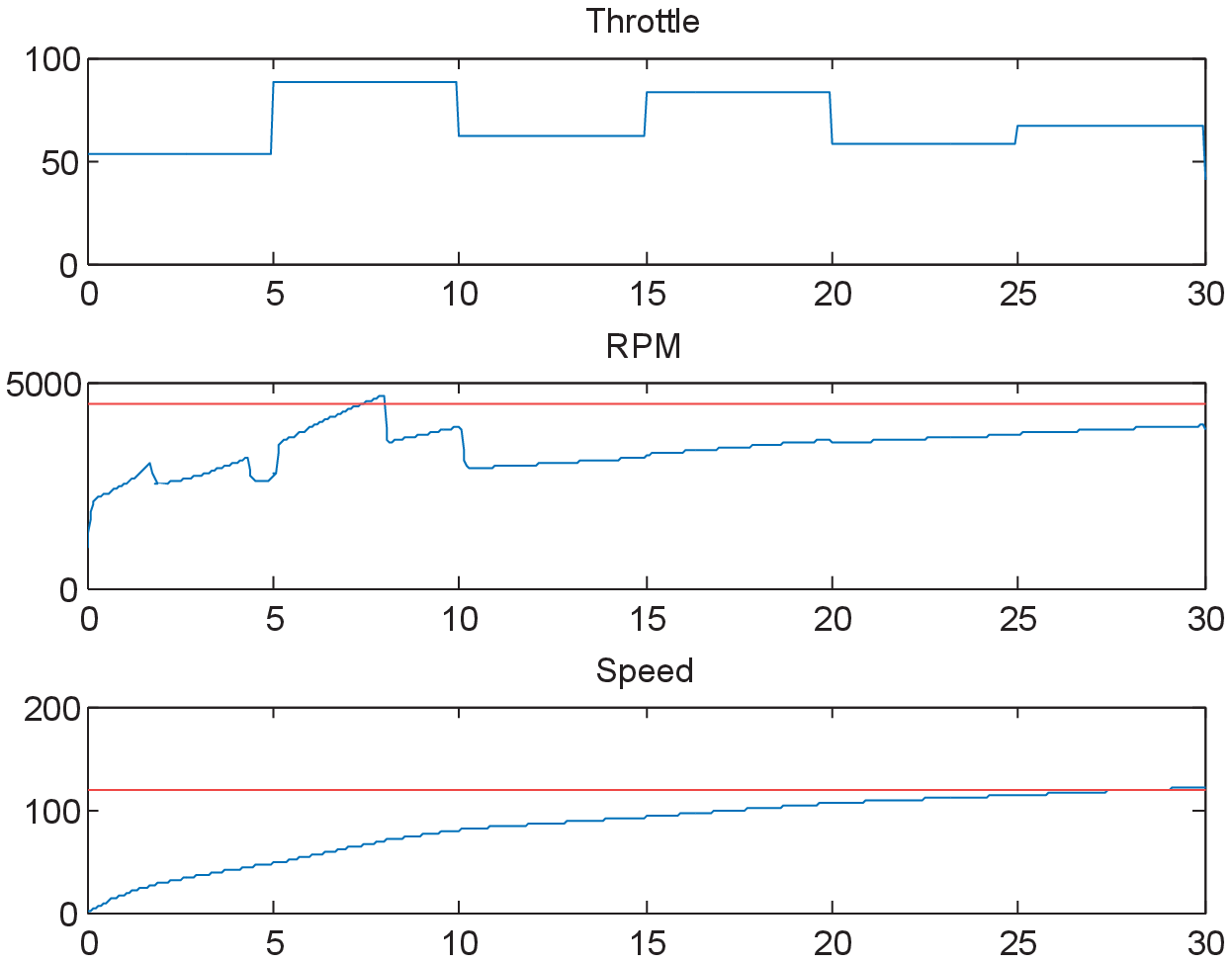} &
\includegraphics[width=7.5cm]{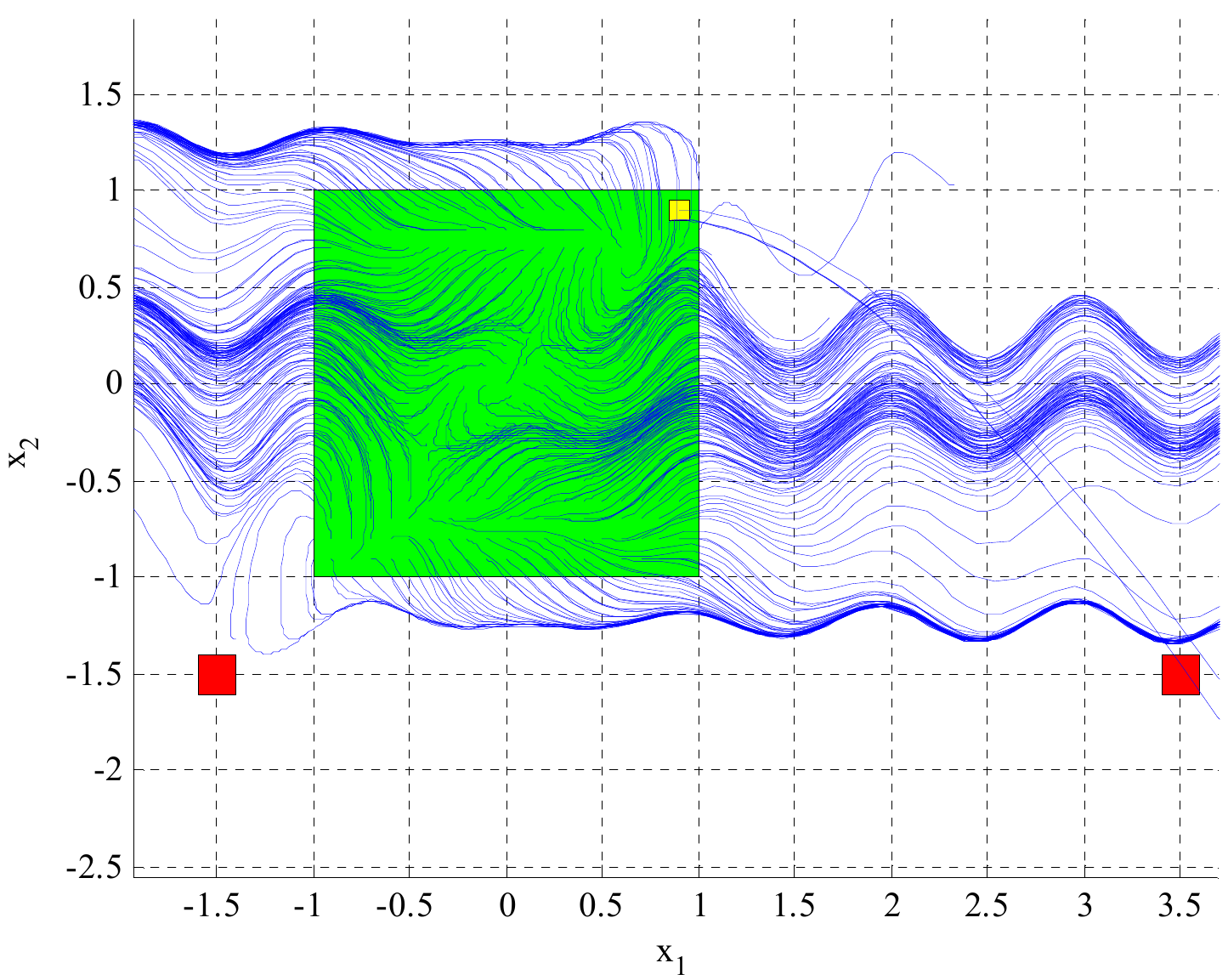} 
\end{tabular}
\caption{\textbf{Left:} Example \ref{exmp:autotrans} (AT): Throttle: A piecewise constant input signal $\InpSig$ parameterized with $\Lambda \in [0,100]^6$ and $t = [0, 5, 10, 15, 20, 25]$. RPM, Speed: The corresponding output signals that falsify the specification ``The vehicle speed $v$ is always under 120km/h or the engine speed $\omega$ is always below 4500RPM." \textbf{Right}: Example \ref{exmp:hs} (HS): Simulated trajectories of the hybrid system containing a trajectory that falsifies the specification ``A trajectory should never pass set $[-1.6, -1.4]^2$ or set $[3.4,3.6] \times [-1.6, -1.4]$". The green square indicates the set of possible initial conditions and the red squares indicate the bad regions which the system should not enter. The yellow region indicates the set of initial conditions where the location on the hybrid system changes. }
\label{fig:falstrajs}
\vspace{-5pt}
\end{figure*}
}{  
\begin{figure}
\centering
\vspace{-10pt}
\includegraphics[width=7.5cm]{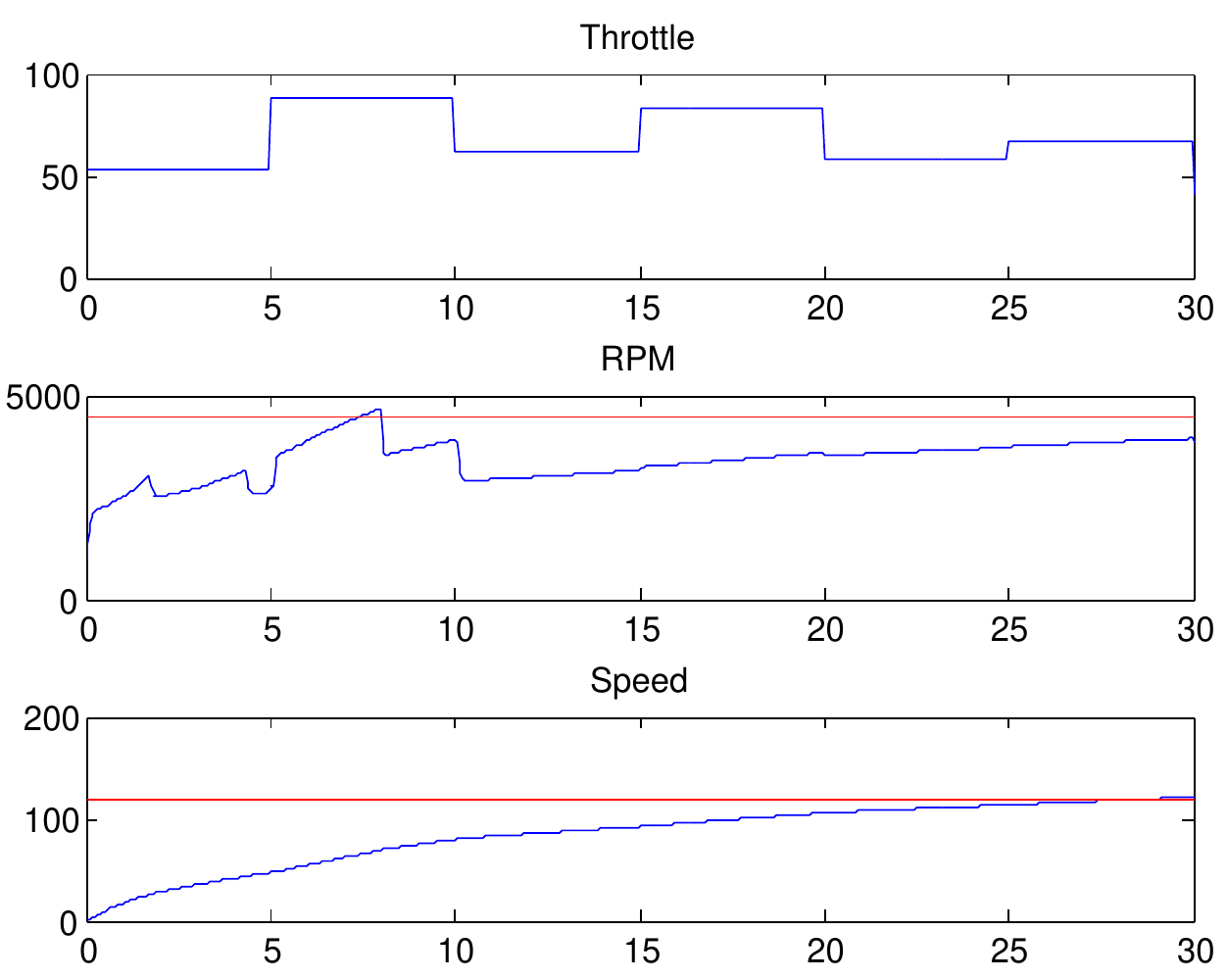}
\caption{Example \ref{exmp:autotrans} (AT): Throttle: A piecewise constant input signal $\InpSig$ parameterized with $\Lambda \in [0,100]^6$ and $t = [0, 5, 10, 15, 20, 25]$. RPM, Speed: The corresponding output signals that falsify the specification ``The vehicle speed $v$ is always under 120km/h or the engine speed $\omega$ is always below 4500RPM."}
\label{fig:falstrajs}
\vspace{-15pt}
\end{figure}	
}

Under Assumption \ref{ass:3}, a system $\Sys$ can be viewed as a function $\SysFun_{\Sys} : \genSS_0 \times \genInpSet \rightarrow \genOut^N \times \Sam$ which takes as an input an initial condition $\gensspt_0 \in \genSS_0$ and an input signal $\InpSig \in \genInpSet$ and it produces as output a signal $\dsig : N \rightarrow \genOut$ (also referred to as {\it trajectory}) and a timing function $\sam : N \rightarrow \preals$.
The only restriction on the timing function $\sam$ is that it must be a monotonic function, i.e., $\sam(i) < \sam(j)$ for $i<j$.
The pair $\tss = (\dsig,\sam)$ is usually referred to as a {\it timed state sequence}, which is a widely accepted model for reasoning about real time systems \cite{alur90realtime}.
A timed state sequence can represent a computer simulated trajectory of a CPS or the sampling process that takes place when we digitally monitor physical systems.
We remark that a timed state sequence can represent both the internal state of the software/hardware (usually through an abstraction) and the state of the physical system.
The set of all timed state sequences of a system $\Sys$ will be denoted by $\Lc(\Sys)$. 
That is, 
\[ \Lc(\Sys) = \{ (\dsig,\sam) \; | \; \exists \gensspt_0 \in \genSS_0 \, . \, \exists  \InpSig \in \genInpSet \, . \, (\dsig,\sam) = \Delta_{\Sys}(\gensspt_0, \InpSig) \}. \]



Our high-level goal is to explore and infer properties that the system $\Sys$ satisfies.
We do so by observing the system response (output signals) to particular input signals and initial conditions. 
We assume that the system designer has partial understanding about the properties that the system satisfies (or does not satisfy) and would like to be able to precisely determine these properties. 
In particular, we assume that the system developer can formalize the system properties in Metric Temporal Logic (MTL) \cite{Koymans90}, where some parameters are unknown.
Such parameters could be unknown threshold values for the continuous state variables of the hybrid system or some unknown real time constraints. 

\ifthenelse{\boolean{TECHREP}}{
{\color{revComments} MTL enables the formalization of complex requirements with respect to both state and time. In addition to propositional logic operators such as conjunction ($\wedge$), disjunction($\vee$) and negation($\neg$), MTL supports temporal operators such as next($X$), until ($\Uc$), release ($\Rc$), always ($\Box$) and eventually ($\Diamond$). Among others, MTL can be utilized to express specifications such as:
\begin{itemize}
\item Safety ($\Box \phi$) : $\phi$ should always hold from this moment on.
\item Liveness ($\Diamond \phi$): $\phi$ should hold at some point in the future (or now).
\item Coverage ($\Diamond \phi_1 \wedge \Diamond \phi_2 \ ... \wedge \Diamond \phi_n$): $\phi_1$ through $\phi_n$ should hold at some point in the future (or now), not necessarily in order or at the same time.
\item Stabilization ($\Diamond \Box \phi$): At some point in the future (or now), $\phi$ should always hold. 
\item Recurrence ($\Box \Diamond \phi$) : At every point time, $\phi$ should hold at some point in the future (or now).  
\end{itemize}
}}{}

{\color{revComments} Another popular formalism for the definition of formal requirements is Signal Temporal Logic (STL) \cite{MalerNickovic04}. 
Since MTL formulas are interpreted over behaviors of the CPS, the results provided in this paper can be directly applied over STL formulas as well.
To enable of elicitation of formal requirements for CPS, tools such as \textsc{ViSpec} \cite{hoxhavispec} may be utilized. }

 

\ifthenelse{\boolean{TECHREP}}{Throughout the paper, we will consider two running examples. The first example consists of an automatic transmission model, and the second, consists of a hybrid non-linear time varying system.}
{
	Throughout the paper, we will consider a running example of an automatic transmission model.
}

\begin{exmp}[AT]
\label{exmp:autotrans}
We consider a slightly modified version of the Automatic Transmission model provided by Mathworks as a Simulink demo\footnote{Available at: \url{http://www.mathworks.com/help/simulink/examples/modeling-an-automatic-transmission-controller.html}}. 
Further details on this example can be found in \cite{ZhaoKH03csm,AbbasFSIG11tecs}. 
The only input $\InpSig$ to the system is the throttle schedule, while the brake schedule is set simply to 0 for the duration of the simulation which is $T = 30 \sec$.
The physical system has two continuous-time state variables which are also its outputs: the speed of the engine $\omega$ (RPM) and the speed of the vehicle $v$, i.e., $\genOut = \Re^2$ and $\gentraj(t) = [\omega(t) \; v(t)]^\intercal$ for all $t \in [0,30]$.
Initially, the vehicle is at rest at time 0, i.e., $\genSS_0 = \{ [0 \; 0]^\intercal \}$ and  $\gensspt_0 = \gentraj(0) = [0 \; 0]^\intercal $.
Therefore, the output trajectories depend only on the input signal $\InpSig$ which models the throttle, i.e., $(\gentraj,\sam) = \SysFun_{\Sys}(\InpSig)$.
The throttle at each point in time can take any value between 0 (fully closed) to 100 (fully open).
Namely, $u(i) \in U=[0,100]$ for each $i \in N$.
The model also contains a Stateflow chart with two concurrently executing Finite State Machines (FSM) with 4 and 3 states, respectively.
The FSM models the logic that controls the switching between the gears in the transmission system.
We remark that the system is deterministic, i.e., under the same input $\InpSig$, we will observe the same output $\gentraj$.
In our previous work \cite{AbbasFSIG11tecs,AnnapureddyLFS11tacas,SankaranarayananF2012hscc}, on such models, we demonstrated how to falsify requirements like: ``The vehicle speed $v$ is always under 120km/h or the engine speed $\omega$ is always below 4500RPM." A falsifying system trajectory appears in Fig. \ref{fig:falstrajs}.
\ifthenelse{\boolean{TECHREP}}{A falsifying system trajectory appears in Fig. \ref{fig:falstrajs} (Left).}\exmend
\end{exmp}

\ifthenelse{\boolean{TECHREP}}{
\begin{exmp}[HS]
\label{exmp:hs}
We consider the hybid time-varying non-linear system presented in Fig. \ref{fig:HS}. 
The output of the system is the state of the system, i.e. $\gentraj(t) = \gensstraj(t)$.
Interesting requirements on this system would be ``A trajectory of the system should never pass through the sets $[-1.6, -1.4]^2$ or $[3.4,3.6] \times [-1.6, -1.4]$". A falsifying system trajectory appears in Fig. \ref{fig:falstrajs} (Right). \exmend


\begin{figure}
	\begin{tikzpicture}
	\usetikzlibrary{calc}
	\usetikzlibrary{positioning}
	\tikzstyle{every node}=[font=\scriptsize, rounded corners=6pt, fill=black!5, draw=black, thick,  minimum size = 1.6cm]
	  \draw (-2.5,2.5) node[] (p1) {\begin{tabular}{ll}  \multicolumn{2}{c}{$S_0$} \\ \\$\dot{x_1} = $&$ x_1(t) - x_2(t) + 0.1t $ \\ $\dot{x_2} = $&$ -x_1(t) \sin(2\pi x_1(t)) + $\\ &$ x_2 \cos(2\pi x_2(t)) + 0.1t $ \end{tabular}};
	   \draw (2.2,2.5) node[] (p2) {\begin{tabular}{l} \multicolumn{1}{c}{$S_1$} \\ \\ $ \dot{x_1} = x_1(t) $ \\ $\dot{x_2} = -x_1(t) + x_2(t) $ \\ \ \end{tabular} };
	 \tikzstyle{every node}=[font=\scriptsize,fill=none]
	     \draw [ bend left, shorten >=1pt,->] (-4.5,4) to node[xshift=25pt, yshift=12pt] {$\gensspt_0 \in [-1,1]^2 \backslash X^U$} ($(p1.north)-(1,0)$);
	     \draw [ bend left, shorten >=1pt,->] ($(p1.north)+(1,0)$) to node[above=1pt] {$\gensspt \in X^U$} (p2);
	     \draw [ bend left, shorten >=1pt,->] (2,4) to node[xshift=10pt, yshift=15pt] {$\gensspt_0 \in  X^U$} ($(p2.north)-(0,0)$);
	\end{tikzpicture}
	\caption{Example \ref{exmp:hs}: Hybid non-linear system with $X^U = [0.85,0.95]^2$ and initial condition $\gensspt_0 \in [-1,1] \times [-1,1]$.}
	\label{fig:HS}
\end{figure}
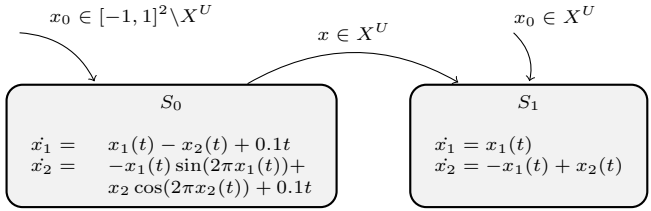

\end{exmp}
}
{}

\subsection{Parameter Mining}

In this work, we provide answers to queries like ``What is the shortest time that $\omega$ can exceed 3250 RPM" or ``For how long can  $\omega$ be below 4500 RPM". We can also answer queries about the relationships between parameters with regard to system falsification. For example, for the specification ``Always the vehicle speed $v$ and engine speed $\omega$ need to be less than parameters $\theta_1,\theta_2$, respectively" we could ask ``If I increase/decrease $\theta_1$ by a specific amount, how much do I have to increase/decrease $\theta_2$ so that I satisfy the specification?".


{\color{revComments}

Formally, we extend and generalize the problem of single parameter mining presented in \cite{YangHF12ictss}. There the problem is defined as follows.


\begin{prob} [MTL 1-Parameter Mining]
Given an MTL formula $\phi[\theta]$ with a single unknown parameter $\theta \in \Theta = [\theta_m,\theta_M]$ and a system $\Sys$, find an optimal range $\Theta^* = [\theta^*_m,\theta^*_M]$ such that for any $\zeta \in \Theta^*$, $\phi[\zeta]$ does not hold on $\Sys$, i.e., $\Sys \not \models \phi[\zeta]$.
\label{prob:mtl:paramOld}
\end{prob}

The extension in the present work is in regards to the number of parameters that can appear in the specification. Formally, it is defined as follows:}

\noindent \begin{prob} [MTL m-Parameter Mining]
Given an  MTL formula $\phi[\vec{\theta}]$ with a vector of $m$ unknown parameters $\vec{\theta} \in \Theta = [\vec{\thetaMin},\vec{\thetaMax}]$ and a system $\Sys$, find the set $\Psi = \{\vec{\theta}^* \in \Theta \ | \ \Sys \not \models \phi[\vec{\theta}^*] \}$. 

\label{prob:mtl:paramExploration}
\end{prob}

That is, the solution to Problem \ref{prob:mtl:paramExploration} is the set $\Psi$ such that for any parameter $\vec{\theta}^*$ in $\Psi$ the specification $\phi[\vec{\theta}^*]$ does not hold on system $\Sys$. 
{\color{revComments}
In the rest of the paper, we refer to $\Psi$ as the parameter falsification domain. 
An approximate solution for Problem \ref{prob:mtl:paramOld} was presented in \cite{YangHF12ictss} for the case where $\theta$ is a scalar. 
In \cite{YangHF12ictss}, the solution to the problem returned a parameter with which the falsifying set can be inferred since the parameter range is one dimensional. 
Here, we provide a solution to Problem \ref{prob:mtl:paramExploration}. 
In the multiple parameter setting, we have a set of possible solutions which we need to explore. 
That is, the solution to the multi-parameter mining problem is in the form of a Pareto front \cite{myers2016response}.

\begin{figure}%
\begin{center}
\includegraphics[width=7.5cm]{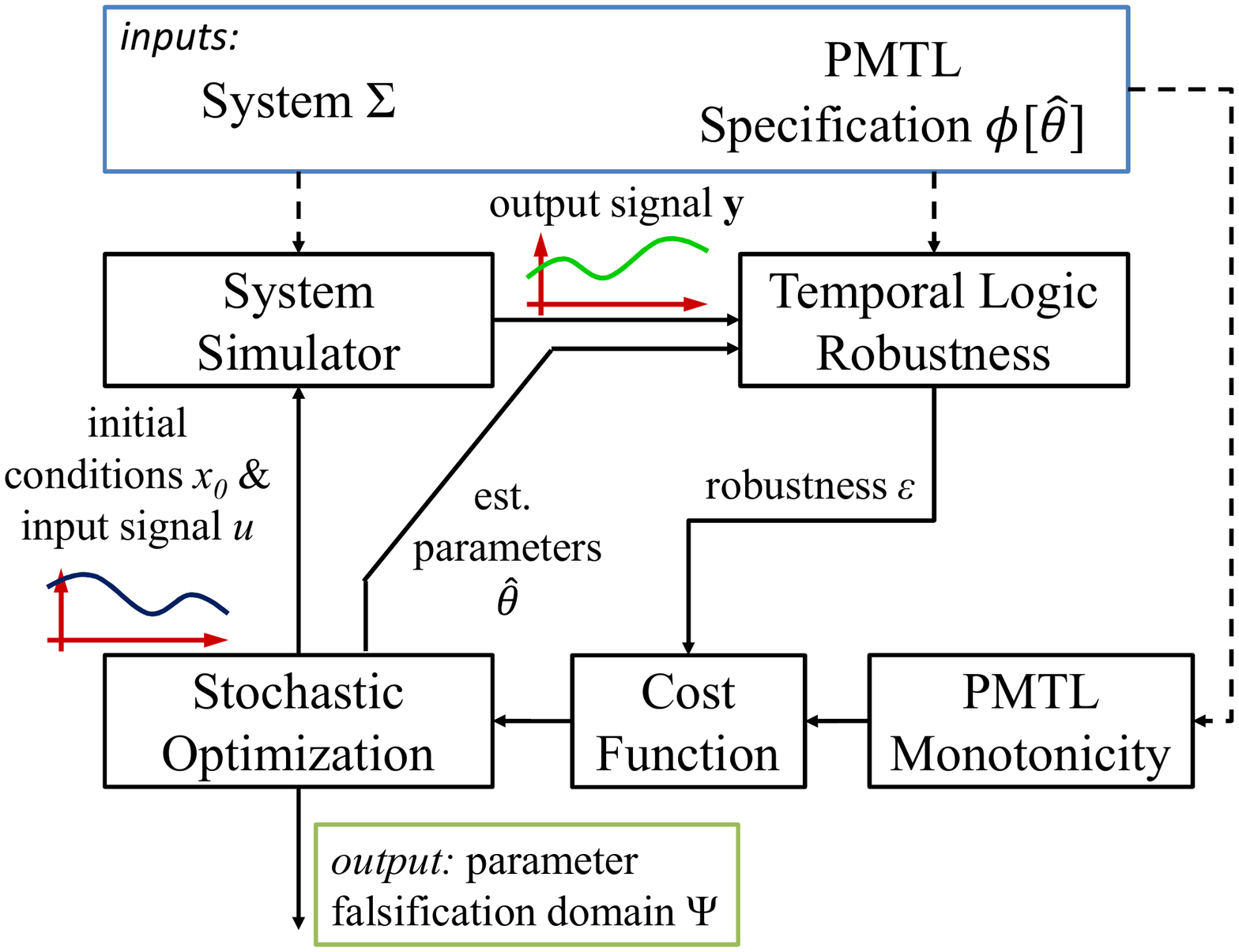} 
\end{center}
\vspace{-5pt} 
\caption{Overview of the solution to Problem \ref{prob:mtl:paramExploration}, the PMTL parameter mining problem for CPS.}
\label{Fig:over:solution}
\vspace{-15pt} 
\end{figure}

 We note that the original observation that the falsification domain problem over \textit{a single system output trace} has the structure of a Pareto front is made in \cite{AsarinDMN12rv}. In this work, we observe that the falsification domain problem over \textit{all system output traces} also has the structure of a Pareto front. Other methods for Pareto front computation have been studied in \cite{legriel2010approximating,deb2001multi}. However, the nature of the problem is significantly different in our case. 
 Here, due to the undecidability of the problem \cite{alur1995algorithmic}, we can only guarantee that a parameter falsifies the specification. 
 \ifthenelse{\boolean{TECHREP}}{It is not the case that we can guarantee that a parameter value satisfies the specification.}{}
 Therefore, the parameter falsification domain is generated strictly by utilizing falsifying behavior. }

Ideally, by solving Problem \ref{prob:mtl:paramExploration}, we would also like to have the property that for any $\vec{\zeta} \in \Theta - \Psi$, $\phi[\vec{\zeta}]$ holds on $\Sys$, i.e., $\Sys \models \phi[\vec{\zeta}]$.
However, even for a given $\vec{\zeta}$, the problem of algorithmically computing whether $\Sys \models \phi[\vec{\zeta}]$ is undecidable for the classes of systems that we consider in this work \cite{alur1995algorithmic}. 



An overview of our proposed solution to Problem \ref{prob:mtl:paramExploration} appears in Fig. \ref{Fig:over:solution}.
Given a model and a MTL specification with one or more parameters, the sampler produces a point $x_0$ from the set of initial conditions, input signal $\InpSig$ and vector of mined parameters $\vec{\theta}$ for the Parametric MTL specification. 
The initial conditions and input signal are passed to the system simulator which returns an execution trace (output trajectory and timing function). 
The trace, in conjunction with the mined parameters, is then analyzed by the MTL robustness analyzer which returns a robustness value.  
The robustness score computed is used by the stochastic sampler to decide on next initial conditions, inputs, and estimated parameters to utilize.
The process terminates after a maximum number of tests or when no improvement on the mined parameters has been made after a number of tests.
As the number of parameters increases, so does the computational complexity of the problem. 
For formulas with more than one parameter, we present an efficient approach in Section \ref{sec:mtl:paramFalsDomain} to explore the parameter falsification domain.

%% file: prelim.tex

\vspace{-8pt}

\section{Robustness of Metric Temporal Logic Formulas}
\label{sec:mtl}

Metric Temporal Logic \cite{Koymans90} enables reasoning over quantitative temporal properties of boolean signals.
In the following, we present MTL in Negation Normal Form (NNF) since this is needed for the presentation of the new results in Section \ref{sec:mtl:param:solution}.
We denote the extended real number line by $\CoRe = \Re\cup\{\pm\infty\}$. 

\begin{defn}[Syntax of MTL in NNF] 
Let $\CoRe$ be the set of truth degree constants, ${\AP}$ be the set of atomic propositions and $\Ic$ be a non-empty non-singular interval of $\CoRe_{\ge 0}$. 
The set $\mtl$ of all well-formed formulas (wff) is inductively defined using the following rules:
\begin{itemize}
\item Terms: True ($\top$), false ($\bot$), all constants $r \in \CoRe$ and atomic propositions $p$, $\neg p$ for $p \in {\AP}$ are terms.
\item Formulas: if $\phi_1$ and $\phi_2$ are terms or formulas, then $\phi_1 \vee \phi_2$, $\phi_1 \wedge \phi_2$, $\phi_1 \Un_\Ic \phi_2$ and $\phi_1 \Rc_\Ic \phi_2$ are formulas.
\end{itemize}
\end{defn}

The atomic propositions in our case label subsets of the output space $\genOut$.
Each atomic proposition is a shorthand for an arithmetic expression of the form $p \equiv  g(\genOutPt) \leq c$, where $g : \genOut \rightarrow \Re$ and $c \in \Re$.
We define an observation map $\Oc : {\AP}\rightarrow \Pc(\genOut)$ such that for each $p \in {\AP}$ the corresponding set is $\Oc(p) = \{ \genOutPt \; | \; g(\genOutPt) \leq c \}  \subseteq \genOut$.

In the above definition, $\Uc_\Ic$ is the timed {\it until} operator and $\Rc_\Ic$ the timed {\it release} operator. 
The subscript $\Ic$ imposes timing constraints on the temporal operators. 
The interval $\Ic$ can be open, half-open or closed, bounded or unbounded, but it must be non-empty ($\Ic \neq \emptyset$) (and, practically speaking, non-singular ($\Ic \neq \{t\}$)).
In the case where $\Ic = [0,+\infty)$, we remove the subscript $\Ic$ from the temporal operators, i.e., we just write $\Uc$ and $\Rc$. 
Also, we can define {\it eventually} ($\Diamond_\Ic \phi \equiv \top \Un_\Ic \phi$) and {\it always} ($\Box_\Ic \phi \equiv \bot \Rc_\Ic \phi$).

Before proceeding to the actual definition of the robust semantics, we introduce some auxiliary notation.
A metric space is a pair $(X,d)$ such that the topology of the set $X$ is induced by a metric $d$.
%
Using a metric $d$, we can define the distance of a point $x \in X$ from a set $S \subseteq X$. 
Intuitively, this distance is the shortest distance from $x$ to all the points in $S$. 
In a similar way, the depth of a point $x$ in a set $S$ is defined to be the shortest distance of $x$ from the boundary of $S$. 
Both the notions of distance and depth play a fundamental role in the definition of the robustness degree. The metrics and distances utilized in this work are covered in more detail in \cite{FainekosP09tcs,AbbasFSIG11tecs}. 

\begin{defn}[Signed Distance] 
\label{def:SignedDistance}
Let $x \in X$ be a point, $S \subseteq X$ be a set and $d$ be a metric on $X$. Then, we define the
Signed Distance from $x$ to $S$ to be 
\[ \mathbf{Dist}_d(x,S) := \left\{ \begin{array}{ll}
- \mathbf{dist}_d(x,S) :=  - \inf\{d(x,y)\;|\;y \in S\} \\ \hspace{2cm}\mbox{ if } x \not \in S \\
\mathbf{depth}_d(x,S) := \mathbf{dist}_d(x,X \backslash S) \\ \hspace{2cm}\mbox{ if } x \in S \\
\end{array} \right. \]
\end{defn}
We utilize the extended definition of the supremum and infimum, i.e., $\sup \emptyset := -\infty$ and $\inf \emptyset := +\infty$.



We define the binary relation $\preceq $ on parameter vectors $\vec{\theta},\vec{\theta}'$ such that $\vec{\theta} \preceq \vec{\theta'} \iff \forall i$,$ \ \theta_i \le  \theta'_i$, where $i$ is the $i^{th}$ entry of the vector. 
MTL formulas are interpreted over timed state sequences $\tss$.
In the past \cite{FainekosP06fates,FainekosP09tcs}, we proposed multi-valued semantics for the MTL where the valuation function on the predicates takes values over the totally ordered set $\CoRe$ according to a metric $d$ operating on the output space $\genOut$. 
We let the valuation function be the depth (or the distance) of the current point of the signal $\dsig(i)$ in a set $\Oc(p)$ labeled by the atomic proposition $p$. 
Intuitively, this distance represents how robust is the point $\dsig(i)$ within set $\Oc(p)$. 
While positive values indicate satisfaction, negative values indicate that the trajectory falsifies the MTL specification.
\ifthenelse{\boolean{TECHREP}}{If this metric is zero, then even the smallest perturbation of the point can drive it inside or outside the set $\Oc(p)$, dramatically affecting membership.}{}This is called a robustness estimate and is formally defined in Definition \ref{def:mitlrob}.

For the purposes of the following discussion, we use the notation $\dle \phi \dri$ to denote the robustness estimate with which the timed state sequence $\tss$ satisfies the specification $\phi$. 
Formally, the valuation function for a given formula $\phi$ is $\dle \phi \dri : \genOut^N \times \Sam \times N \rightarrow \overline{\Re}$.
In the definition below, we also use the following notation : for $Q \subseteq R$, the {\it preimage} of $Q$ under $\sam$ is defined as : $\sam^{-1}(Q) := \{i \in N \; | \; \sam(i) \in Q \}$. 

\ifthenelse{\boolean{TECHREP}}{

\begin{figure*}
\centering
\vspace{-10pt}
\begin{tabular}{cc}
\includegraphics[width=7.5cm]{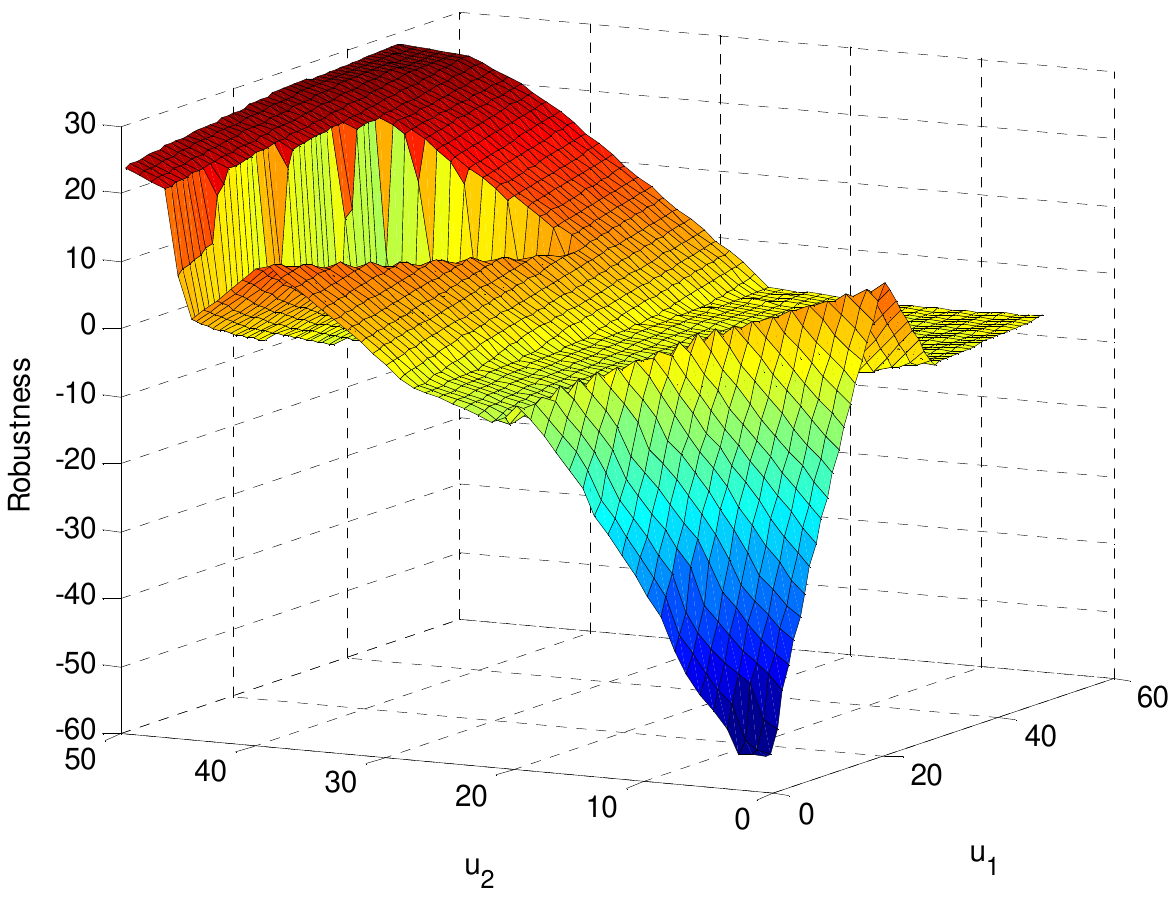} &
\includegraphics[width=7.5cm]{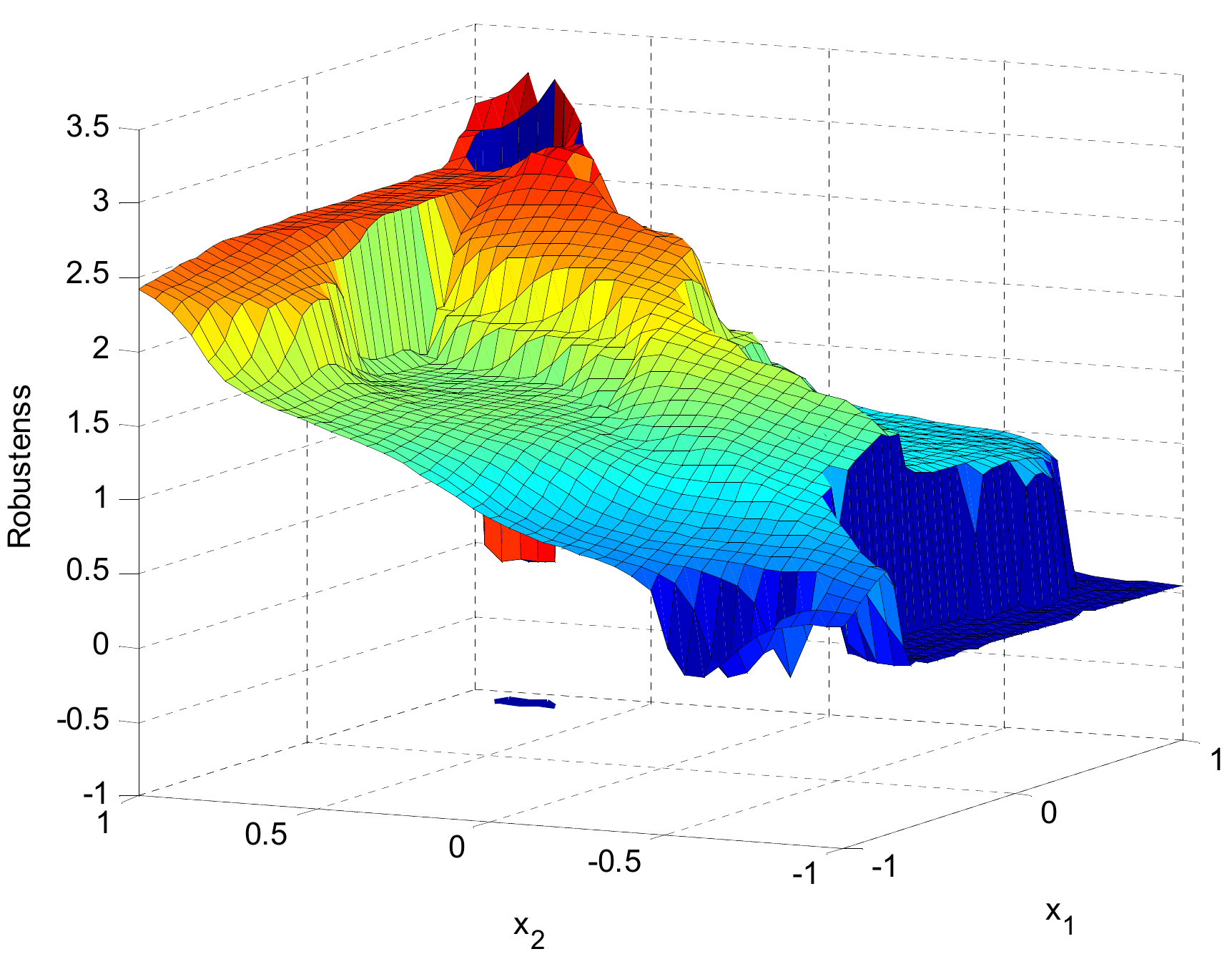} 
\end{tabular}
\caption{Robustness estimate landscape for system specifications. \textbf{Left:} Example \ref{exmp:autotrans} (AT): $\phi_{AT} = \neg (\Diamond_{[0,30]} (v > 100) \wedge \Box (\omega \le 4500)) \wedge \neg \Diamond_{[10,40]}\Box_{[0,5]}(60<v \le 80) \wedge \neg \Diamond_{[50,60]}\Box_{[0,3]}(v \le 60)$. The input signal to the system is generated by linearly interpolating control points $u_1$, $u_2$ at time 0 and 60, respectively, for the throttle input $u$. That is, $u(t) = \frac{60-t}{60}u_1 + \frac{t}{60}u_2$.; \textbf{Right:} Example \ref{exmp:hs} (HS): $\phi_{HS} = \Box_{[0,2]} \neg a \wedge \Box_{[0,2]} \neg b $, where $\Oc(a) = [-1.6, -1.4]^2$ and $\Oc(b) = [3.4,3.6] \times [-1.6, -1.4]$. Here $x_1$ and $x_2$ are initial conditions for the hybrid system.}
\label{fig:roblandscape}
\vspace{-5pt}
\end{figure*} 
}
{

\begin{figure}
\centering
\vspace{-5pt}
\includegraphics[width=6cm]{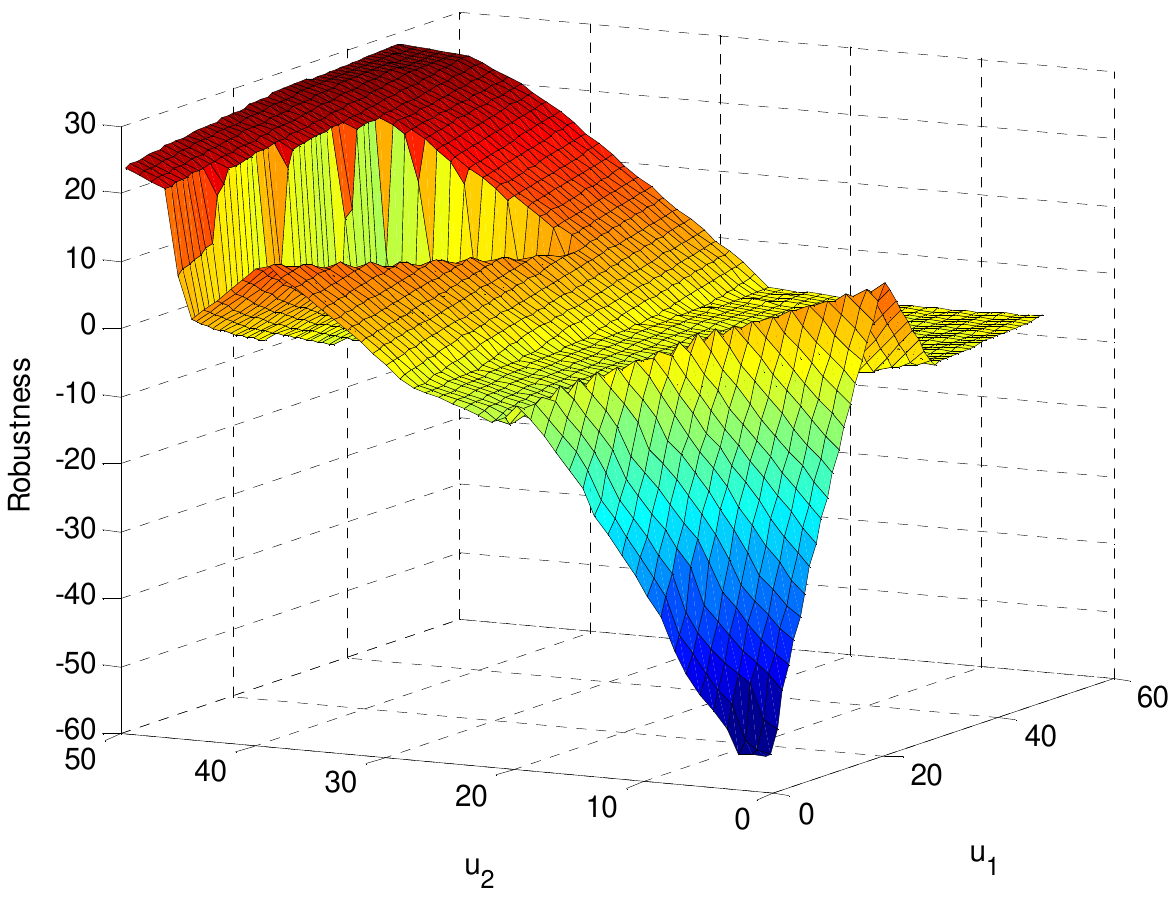} 
\vspace{-5pt}
\caption{Robustness estimate landscape for system specifications. Example \ref{exmp:autotrans} (AT): $\phi_{AT} = \neg (\Diamond_{[0,30]} (v > 100) \wedge \Box (\omega \le 4500)) \wedge \neg \Diamond_{[10,40]}\Box_{[0,5]}(60<v \le 80) \wedge \neg (\Diamond_{[50,60]}\Box_{[0,3]}(v \le 60))$. The input signal to the system is generated by linearly interpolating control points $u_1$, $u_2$ at time 0 and 60, respectively, for the throttle input $u$. That is, $u(t) = \frac{60-t}{60}u_1 + \frac{t}{60}u_2$.}
\label{fig:roblandscape}
\vspace{-20pt}
\end{figure} 

}

\begin{defn} [Robustness Estimate \cite{FainekosP09tcs}] 
Let $\tss =$ $(\dsig,\sam) \in \Sig$, $r \in \CoRe$ and $i,j,k \in N$, then the robustness estimate of any formula MTL $\phi$ with respect to $\tss$ is recursively defined as follows
\begin{gather*} 
\allowdisplaybreaks
\dle \top \dri (\tss,i)  := +\infty  \qquad \dle \bot \dri (\tss,i) := -\infty \\
\dle p \dri (\tss,i)  := \Dist{d}(\dsig(i),\Oc(p))    \\ \dle \neg p \dri (\tss,i)  := -\Dist{d}(\dsig(i),\Oc(p))  \\
\dle \phi_1 \vee \phi_2 \dri (\tss,i)  := \max(\dle \phi_1 \dri (\tss,i), \dle \phi_2 \dri (\tss,i)) \\
\dle \phi_1 \wedge \phi_2 \dri (\tss,i)  := \min(\dle \phi_1 \dri (\tss,i), \dle \phi_2 \dri (\tss,i)) \\ 
\dle \phi_1 \Un_\Ic \phi_2 \dri (\tss,i)  :=  \nonumber \\ \sup_{j \in \sam^{-1}(\sam(i)+\Ic)} \bigl( \min ( \dle \phi_2 \dri (\tss,j),  \inf_{i\leq k <j} \dle \phi_1 \dri (\tss,k)) \bigr)  \label{def3until} \\ 
\dle \phi_1 \Rc_\Ic \phi_2 \dri (\tss,i)  :=  \nonumber \\ \inf_{j \in \sam^{-1}(\sam(i)+\Ic)} \bigl( \max ( \dle \phi_2 \dri (\tss,j), \sup_{i\leq k < j} \dle \phi_1 \dri (\tss,k) ) \bigr) 
\end{gather*}
\label{def:mitlrob}
\end{defn}
Recall that we use the extended definition of supremum and infimum.
When $i=0$, then we write $\dle \phi \dri (\tss)$.

The robustness of an MTL formula with respect to a timed state sequence can be computed using several existing algorithms \cite{FainekosP09tcs,FainekosSUY12acc,DonzeM10formats}.

If we consider the robustness estimate over systems, the resulting robustness landscape can be both non-linear and non-convex.
\ifthenelse{\boolean{TECHREP}}{In Fig. \ref{fig:roblandscape} we present the robustness landscape for the two running examples, namely Examples \ref{exmp:autotrans} (AT) and \ref{exmp:hs} (HS), on two specifications.}{In Fig. \ref{fig:roblandscape} we present the robustness landscape for the automotive running example.}






\vspace{-4pt}

\section{Parametric Metric Temporal Logic over Signals}

\vspace{-4pt}

In many cases, it is important to be able to describe an MTL specification with unknown parameters and then, infer the parameters that make the specification false.
In \cite{AsarinDMN12rv}, Asarin et al. introduced Parametric Signal Temporal Logic (PSTL) and presented two algorithms for computing approximations for parameters over a given signal.
Here, we review some of the results in \cite{AsarinDMN12rv} while adapting them in the notation and formalism that we use in this paper.

\begin{defn}[Syntax of Parametric MTL]
Let $\vec{\theta}$ be a vector of parameters. The set of all well formed Parametric MTL (PMTL) formulas is the set of all well formed MTL formulas where for all $i$, $\theta_i$ either appears in an arithmetic expression, i.e., $p[\theta_i] \equiv g(\genOutPt) \leq \theta_i$, or in the timing constraint of a temporal operator, i.e., $\Ic[\theta_i]$.  
\end{defn}

We will denote a PMTL formula $\phi$ with parameters $\vec{\theta}$ by $\phi[\vec{\theta}]$.
Given a vector of parameters $
\vec{\theta} \in \Theta$, then the formula $\phi[\vec{\theta}]$ is an MTL formula. There is an implicit mapping from the vector of parameters $\vec{\theta}$ to the corresponding arithmetic expressions and temporal operators in the MTL formula.

Since the valuation function of an MTL formula is a composition of minimum and maximum operations quantified over time intervals, a formula $\phi[\theta]$, when $\theta$ is a scalar, is always monotonic with respect to $\theta$ under certain conditions. Similarly, when $\vec{\theta}$ is a vector, then the valuation function is monotonic with respect to a priority function $f(\vec{\theta})$. {\color{revComments} In general, determining the monotonicity of PMTL formulas is undecidable \cite{jin2013mining}}. The priority function will enable the system engineer to prioritize the optimization of some parameters over others by defining specific weights, or setting an optimization strategy such as optimizing the minimum, maximum, or norm of all parameters. The priority function will be defined in detail in the next section.

\ifthenelse{\boolean{TECHREP}}{}{}
\begin{figure*}
\centering
\vspace{-10pt}
\begin{tabular}{cc}
\includegraphics[width=6.2cm]{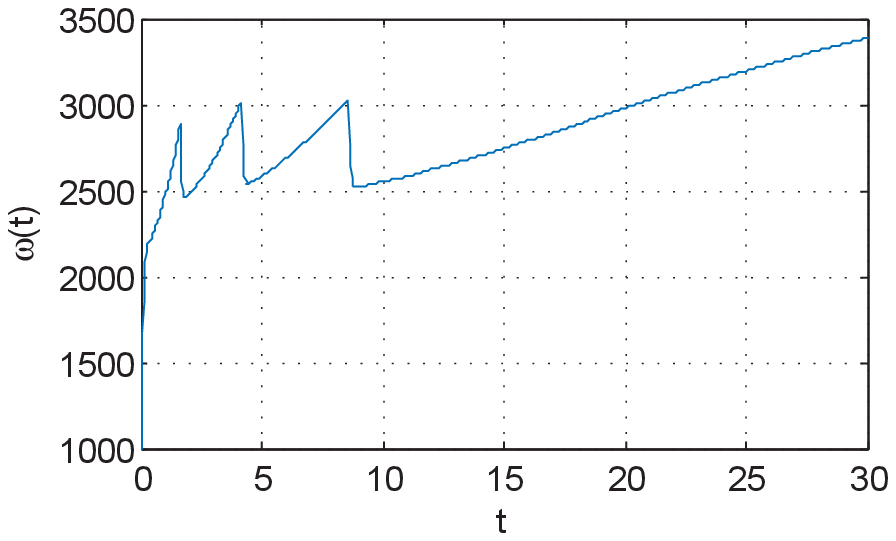} &
\includegraphics[width=6.1cm]{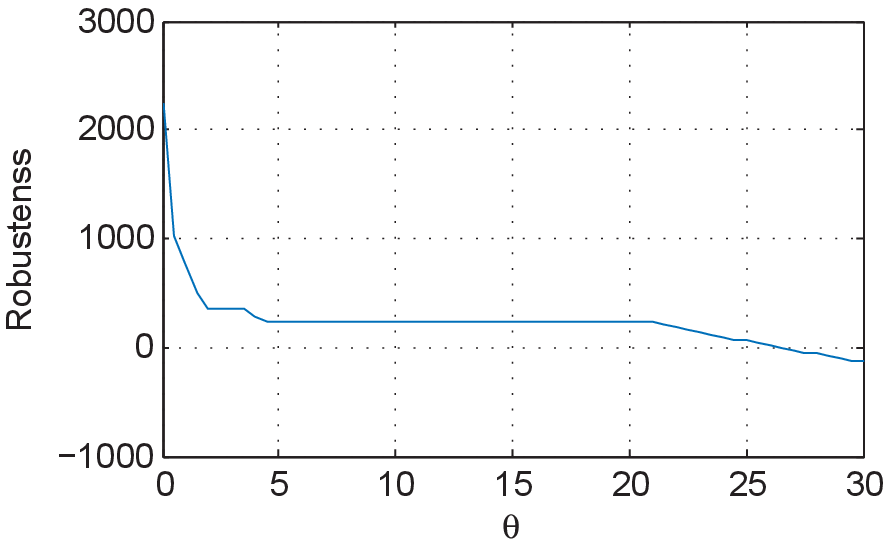} 
\end{tabular}
\vspace{-5pt}
\caption{Example \ref{exmp:rob:always}. Left: Engine speed $\omega(t)$ for constant throttle $u(t) = 50$. Right: The robustness estimate of the specification $\Box_{[0,\theta]} (\omega \leq 3250)$ with respect to $\theta$.}
\label{fig:rob:exmp:always}
\vspace{-10pt}
\end{figure*}

In the following, we present monotonicity results for single and multiple parameter PMTL formulas. {\color{revComments} We note that the monotonicity results apply to a subset of PMTL}.
\ifthenelse{\boolean{TECHREP}}{
\subsection{Single parameter PMTL formulas}
}{}

The first example presented shows how monotonicity appears in the timing requirements of PMTL formulas.

\begin{exmp}[AT]
\label{exmp:rob:always}

Consider the PMTL formula $\phi[\theta] = \Box_{[0,\theta]}  p$ where $p \equiv (\omega \leq 3250)$.
Given a timed state sequence $\tss = (\dsig,\sam)$ with $\sam(0) = 0$, for $\theta_1 \leq \theta_2$, we have: 
\begin{gather*}
[0,\theta_1] \subseteq [0,\theta_2] \implies \sam^{-1}([0,\theta_1]) \subseteq \sam^{-1}([0,\theta_2]).
\end{gather*}
Therefore, by Definitions (\ref{def:SignedDistance}) and (\ref{def:mitlrob}) we have
\begin{gather*}
\dle \phi[\theta_1] \dri (\tss) = \inf_{i \in \sam^{-1}([0,\theta_1])} (-\Dist{d}(\dsig(i),\Oc(p))) \\
 \geq \inf_{i \in \sam^{-1}([0,\theta_2])} (-\Dist{d}(\dsig(i),\Oc(p))) = \dle \phi[\theta_2] \dri (\tss).
\end{gather*}
That is, the function $\dle \phi[\theta] \dri(\tss)$ is non-increasing with $\theta$. Intuitively, this relationship holds since by extending the value of $\theta$ in $\phi[\theta]$, it becomes just as or more difficult to satisfy the specification.
See Fig.~\ref{fig:rob:exmp:always} for an example using an output trajectory from the system in Example~\ref{exmp:autotrans}. $\exmend$ 
\end{exmp}

The aforementioned example is formalized by the following monotonicity results.
{ \color{revComments}
\begin{lem}[Extended from \cite{YangHF12ictss}]
Consider a PMTL formula $\phi[\theta]$ such that it contains one or more subformulas $\phi_1 Op_{\Ic[\theta]} \phi_2$ where $Op \in \{ \Uc, \Rc\}$. 
Then, given a timed state sequence $\tss = (\dsig,\sam)$, for $\theta_1$, $\theta_2$ $\in \CoRe_{\ge 0}$, such that $\theta_1 \leq \theta_2$, and for $i \in N$, we have:
\begin{enumerate}
\item if for all such subformulas, we have (i) $Op = \Uc$ and $\sup \Ic(\theta) = \theta$ or (ii) $Op = \Rc$ and $\inf \Ic(\theta) = \theta$, then $\dle \phi[\theta_1] \dri (\tss,i) \leq \dle \phi[\theta_2] \dri (\tss,i)$, i.e., the function $\dle \phi[\theta] \dri (\tss,i)$ is non-decreasing with respect to $\theta$.
\item if for all such subformulas, we have (i) $Op = \Rc$ and $\sup \Ic(\theta) = \theta$ or (ii) $Op = \Uc$ and $\inf \Ic(\theta) = \theta$, then $\dle \phi[\theta_1] \dri (\tss,i) \geq \dle \phi[\theta_2] \dri (\tss,i)$, i.e., the function $\dle \phi[\theta] \dri (\tss,i)$ is non-increasing with respect to $\theta$.
\end{enumerate}
\label{lem:timeOperMonDec}
\end{lem}


}

\ifthenelse{\boolean{TECHREP}}{

\begin{proof} [sketch]
Without loss of generality, we will prove only case (i) of Lemma \ref{lem:timeOperMonDec}.1. Case (ii) is symmetric with respect to the temporal operator and Lemma \ref{lem:timeOperMonDec}.2 is symmetric in terms of monotonicity. The proof is by induction on the structure of the formula and it is similar to the proofs that appeared in \cite{FainekosP09tcs}.


For completeness of the presentation, we consider the case $\dle \phi_1 \Un_{\langle \alpha, \theta \rangle} \phi_2 \dri (\tss,i)$, where $\langle \in \{ [, ( \}$ and $\rangle \in \{ ], ) \}$.
The other cases are either similar or they are based on the monotonicity of the $\max$ and $\min$ operators.
We remark that the $\max$ and $\min$ operators preserve monotonicity.
Let $\theta_1 \leq \theta_2$, then we want to show that:
\begin{gather}
\dle \phi_1 \Un_{\langle \alpha, \theta_1 \rangle} \phi_2 \dri (\tss,i) \leq 
\dle \phi_1 \Un_{\langle \alpha, \theta_2 \rangle} \phi_2 \dri (\tss,i) \label{lemma1proof}
\end{gather}

To show that (\ref{lemma1proof}) holds, we utilize the robust semantics for MTL given in Definition \ref{def:mitlrob} and observe that: 
\begin{gather*}
\dle \phi_1 \Un_{\langle \alpha, \theta_2 \rangle} \phi_2 \dri (\tss,i) =  \\ \sup_{j \in \sam^{-1}(\sam(i)+\langle \alpha, \theta_2 \rangle)} \bigl( \min ( \dle \phi_2 \dri (\tss,j),  \inf_{i\leq k <j} \dle \phi_1 \dri (\tss,k)) \bigr) = \\ \max \left(\sup_{j \in \sam^{-1}(\sam(i)+\langle \alpha, \theta_1 \rangle)} \bigl( \min ( \dle \phi_2 \dri (\tss,j),  \inf_{i\leq k <j} \dle \phi_1 \dri (\tss,k)) \bigr), \right.
\\
\left. \vphantom{\int_1^2}\sup_{j \in \sam^{-1}(\sam(i)+\overline{\langle} \theta_1, \theta_2 \rangle)} \bigl( \min ( \dle \phi_2 \dri (\tss,j),  \inf_{i\leq k <j} \dle \phi_1 \dri (\tss,k)) \bigr) \vphantom{\int_1^2} \right) = \\
\max \left(\dle \phi_1 \Un_{\langle \alpha, \theta_1 \rangle} \phi_2 \dri (\tss,i), \dle \phi_1 \Un_{\overline{\langle} \theta_1, \theta_2 \rangle} \phi_2 \dri (\tss,i) \right ) \ge \\
\dle \phi_1 \Un_{\langle \alpha, \theta_1 \rangle} \phi_2 \dri (\tss,i) 
\end{gather*}

where $\overline{\langle} \in \{ [, ( \}$ such that ${\langle \alpha, \theta_1 \rangle} \cap \overline{\langle} \theta_1, \theta_2 \rangle = \emptyset$ and ${\langle \alpha, \theta_1 \rangle} \cup \overline{\langle} \theta_1, \theta_2 \rangle = \langle \alpha, \theta_2 \rangle$. $\hfill {\qed}$

\end{proof}
}
{
The proof for Lemma \ref{lem:timeOperMonDec} is presented in the extended version of the paper in \cite{Hoxha2016MiningExtendedVer}.
}
\color{revComments} Note that Lemma \ref{lem:timeOperMonDec} allows for the repetition of a parameter in a PMTL formula. For example, consider the specification $\phi = \Box_{[\theta,5]}a \wedge \Diamond_{[0,\theta]}b \equiv \bot \mathcal{R}_{[\theta,5]}a \wedge \top \mathcal{U}_{[0,\theta]}b$. In this case, $\phi$ satisfies the conditions in Lemma \ref{lem:timeOperMonDec}. Thus, from Lemma \ref{lem:timeOperMonDec} we know that for two values $\theta_1$ and $\theta_2$ where $\theta_1 \leq \theta_2$: 
\begin{gather*}
\dle \Box_{[\theta_1,5]}a \wedge \Diamond_{[0,\theta_1]}b \dri (\tss,i) \leq \dle \Box_{[\theta_2,5]}a \wedge \Diamond_{[0,\theta_2]}b \dri (\tss,i)
\end{gather*}

\ifthenelse{\boolean{TECHREP}}{}{}
\begin{figure*}
\centering
\begin{tabular}{cc}
\includegraphics[width=6.2cm]{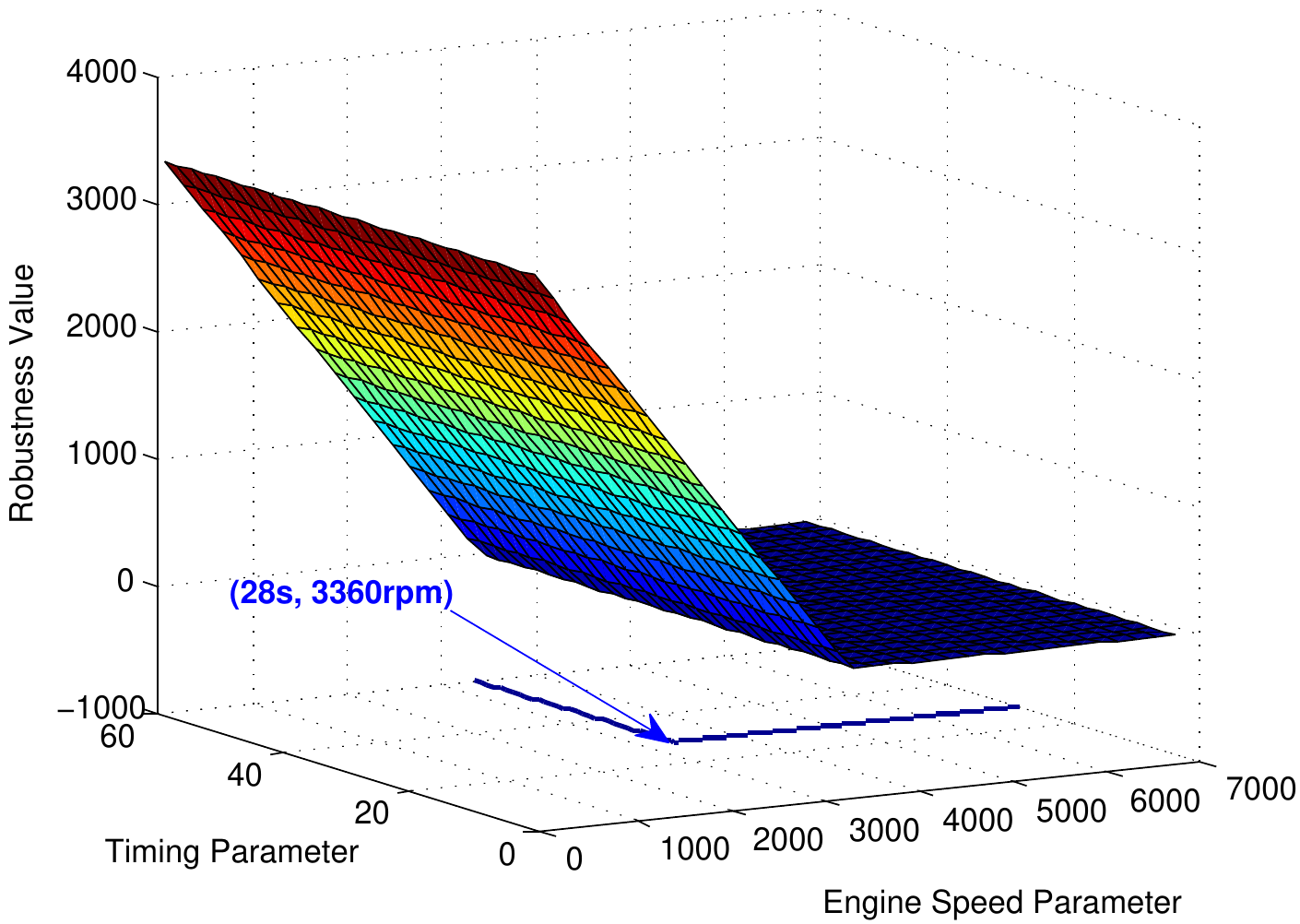}  &
\includegraphics[width=6.2cm]{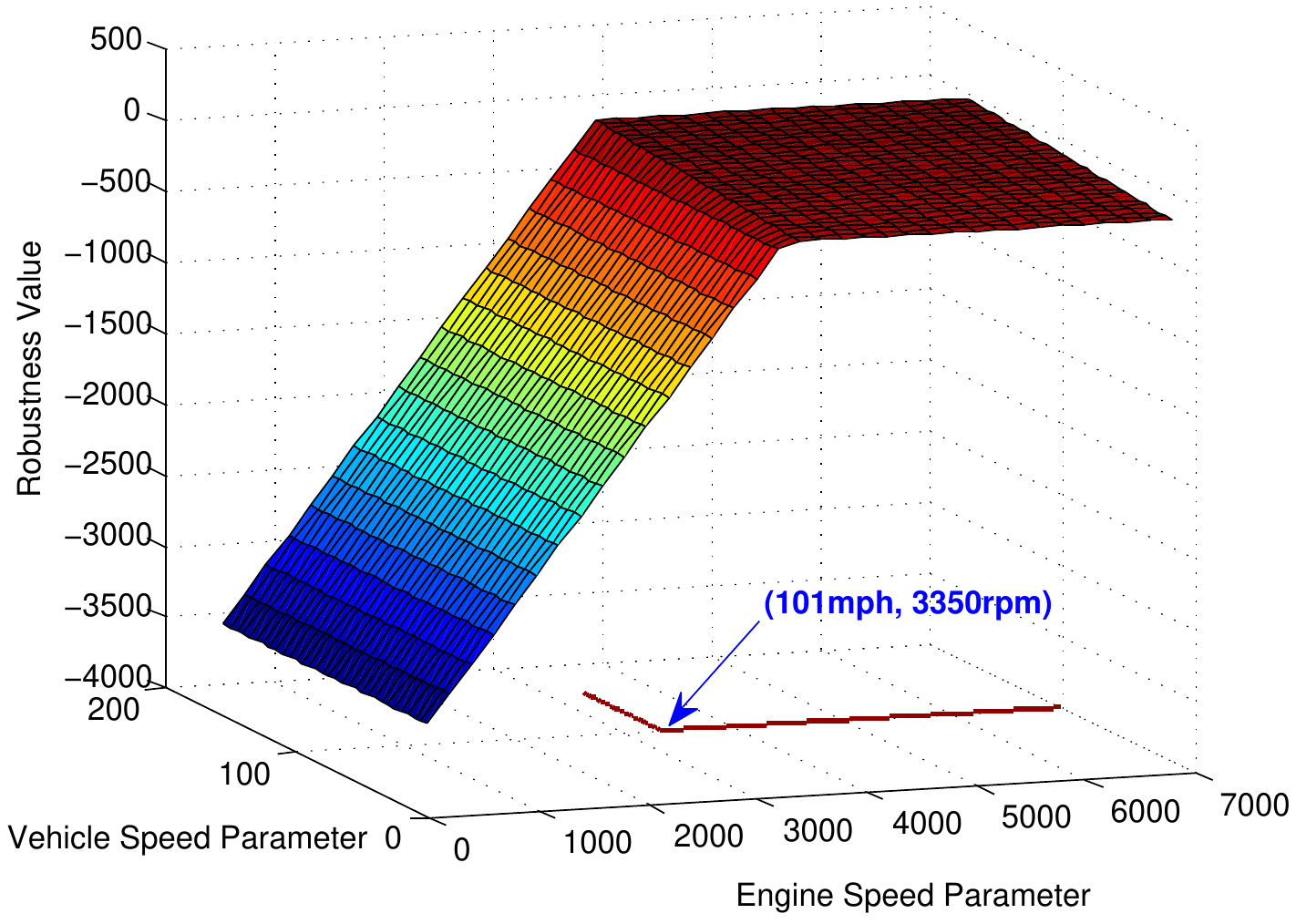}
\end{tabular}
\caption{ \textbf{Left:} Example \ref{exmp:rob:multParam1}: Robustness estimate landscape for varying parameters for engine and vehicle speed for constant throttle $u(t) = 50$. \textbf{Right}: Example \ref{exmp:rob:multParam2}: Robustness landscape for varying parameters for timing parameter and engine speed for constant throttle $u(t) = 50$. In both figures, the contour line shows the intersection of the robustness landscape with the zero level set.}
\label{fig:rob:exmp:robLandscape}
\vspace{-12pt}
\end{figure*}

In the following, we derive similar results for the case where the parameter appears in the numerical expression of the atomic proposition.

\begin{lem}[Extended from \cite{YangHF12ictss}]
Consider a PMTL formula $\phi[\theta]$ with a single parameter $\theta$ such that it contains parametric atomic propositions $p_1[\theta]...p_n[\theta]$ in one or more subformulas.
Then, given a timed state sequence $\tss = (\dsig,\sam)$, for $\theta_1$, $\theta_2$ $\in \CoRe_{\ge 0}$, such that $\theta_1 \leq \theta_2$, and for $i \in N$, we have:
\begin{itemize}
\item if $\forall j.p_j[\theta] \equiv g_j(x) \leq \theta$, then $\dle \phi[\theta_1] \dri (\tss,i) \leq \dle \phi[\theta_2] \dri (\tss,i)$, i.e., the function $\dle \phi[\theta] \dri (\tss,i)$ is non-decreasing with respect to $\theta$, and  
\item if $\forall j.p_j[\theta] \equiv g_j(x) \geq \theta$, then $\dle \phi[\theta_1] \dri (\tss,i) \geq \dle \phi[\theta_2] \dri (\tss,i)$, i.e., the function $\dle \phi[\theta] \dri (\tss,i)$ is non-increasing with respect to $\theta$.
\end{itemize}
\label{lem:AtmPropMon}
\end{lem}

\ifthenelse{\boolean{TECHREP}}{

\begin{proof} [sketch]
The proof is by induction on the structure of the formula  and it is similar to the proofs that appeared in \cite{FainekosP09tcs}.
For completeness of the presentation, we consider the base case $\dle p[\theta] \dri (\tss,i)$.
Let $\theta_1 \leq \theta_2$, then $\Oc(p[\theta_1]) \subseteq  \Oc(p[\theta_2])$.
We will only present the case for which $\dsig(i) \not \in \Oc(p[\theta_2])$.
We have:
\begin{gather*}
\Oc(p_j[\theta_1]) \subseteq  \Oc(p_j[\theta_2]) \implies \\ \dist{d}(\dsig(i),\Oc(p_j[\theta_1])) \geq \dist{d}(\dsig(i),\Oc(p_j[\theta_2])) \implies \\
\Dist{d}(\dsig(i),\Oc(p_j[\theta_1])) \leq \Dist{d}(\dsig(i),\Oc(p_j[\theta_2]))  \implies \\ \dle p_j[\theta_1] \dri (\tss,i) \leq \dle p_j[\theta_2] \dri(\tss,i)  \tag*{\qed }
\end{gather*}
\end{proof}
}
{
The proof for Lemma \ref{lem:AtmPropMon} is presented in the extended version of the paper in \cite{Hoxha2016MiningExtendedVer}.	
}

\ifthenelse{\boolean{TECHREP}}{
\vspace{-10pt}
\subsection {Multiple parameter PMTL formulas}
}{}

Next, we extend the result for multiple parameters.

\begin{exmp}[AT]
\label{exmp:rob:multParam1}
Consider the PMTL formula $\phi[\vec{\theta}]=\neg  (\Diamond_{[0,\theta_1]} $ $ q \wedge \Box p[\theta_2])$  where $\vec{\theta} = [\theta_1,\theta_2]^\intercal$, $p[\theta_2] \equiv (\omega \leq \theta_2)$ and $q \equiv (v \geq 100)$.
Given a timed state sequence $\tss = (\dsig,\sam)$ with $\sam(0) = 0$, for two vectors of parameters $\vec{\theta},\vec{\theta}' \in \Re^2$ where $\vec{\theta} \preceq \vec{\theta'}$, we have: for all $i$,
\begin{gather}
\theta_2 \leq \theta'_2 \implies \Oc(p[\theta_2]) \subseteq \Oc(p[\theta'_2]) \implies \nonumber\\ 
\Dist{d}(\dsig(i),\Oc(p[\theta_2])) \leq \Dist{d}(\dsig(i),\Oc(p[\theta'_2])) \implies \nonumber\\  -\Dist{d}(\dsig(i),\Oc(p[\theta_2])) \geq -\Dist{d}(\dsig(i),\Oc(p[\theta'_2])) 
\label{eq:1}
\end{gather}
\begin{gather}
\theta_1 \leq \theta'_1 \implies [0,\theta_1] \subseteq [0,\theta'_1] \implies \nonumber\\ \sam^{-1}([0,\theta_1]) \subseteq \sam^{-1}([0,\theta'_1]) 
\label{eq:2}
\end{gather}
Therefore, by (\ref{eq:1}) and (\ref{eq:2}) we obtain:
\begin{gather}
 \dle \phi[\vec{\theta}] \dri (\tss) = \inf_{i \in \sam^{-1}([0,\theta_1])}  (-\Dist{d}(\dsig(i),\Oc(p[\theta_2]))) \nonumber \\ 
 \overset{(2)}{\geq} \inf_{i \in \sam^{-1}([0,\theta_1])}  (-\Dist{d}(\dsig(i),\Oc(p[\theta'_2]))) \nonumber \\ 
 \overset{(3)}{\geq} \inf_{i \in \sam^{-1}([0,\theta'_1])}  (-\Dist{d}(\dsig(i),\Oc(p[\theta'_2]))) = \dle \phi[\vec{\theta'}] \dri (\tss) \nonumber
\end{gather}
That is, the function $\dle \phi[\vec{\theta}] \dri(\tss)$ is non-increasing for all $\vec{\theta}$ for which the relation $\preceq$ holds.
 $\hfill \blacktriangle$
\end{exmp}

\begin{exmp}[AT]
\label{exmp:rob:multParam2}
Consider the PMTL formula $\phi[\vec{\theta}]$ $=$ $ \square (p[\theta_1]$ $ \wedge q[\theta_2])$  where $p[\theta_1] \equiv (v \leq \theta_1)$ and $q[\theta_2] \equiv (\omega \leq \theta_2)$.
Given a timed state sequence $\tss = (\dsig,\sam)$ with $\sam(0) = 0$, for two vectors of parameters $\vec{\theta},\vec{\theta}'$ where $\vec{\theta} \preceq \vec{\theta}'$, we have: 
\vspace{-12pt}
\begin{gather*}
\Oc(p[\theta_1]) \subseteq  \Oc(p[\theta'_1]) \implies \\
\Dist{d}(\Oc(p[\theta_1])) \leq  \Dist{d}(\Oc(p[\theta'_1])) \implies \\
\dle p[\theta_1] \dri (\tss,i) \leq \dle p[\theta'_1]  \dri(\tss,i)  
 \\ \text{ and }  \\ \Oc(q[\theta_2]) \subseteq  \Oc(q[\theta'_2]) \implies \\ 
\Dist{d}(\Oc(p[\theta_2])) \leq \Dist{d}(\Oc(p[\theta'_2])) \implies \\
\dle q[\theta_2] \dri (\tss,i) \leq \dle q[\theta'_2] \dri(\tss,i)
\end{gather*}
Therefore, $\dle \phi[\vec{\theta}] \dri (\tss) \le \dle \phi[\vec{\theta}'] \dri (\tss)$. That is, the function $\dle \phi[\vec{\theta}] \dri(\tss)$ is non-decreasing for all $\vec{\theta}$ for which the relation $\preceq$ holds.
Figure~\ref{fig:rob:exmp:robLandscape} presents the robustness landscape of two parameters over constant input.
$\hfill \blacktriangle$
\end{exmp}

Now we may state the main monotonicity theorem for multiple parameters. 
\ifthenelse{\boolean{TECHREP}}{We remark that for convenience we define the parametric subformulas over all the possible parameters even though only some of them are used in each subformula.}

\begin{theorem}
Consider a PMTL formula $\psi[\vec{\theta}]$, where $\vec{\theta}$ is a vector of parameters, such that $\psi[\vec{\theta}]$ contains temporal subformulas $\phi[\vec{\theta}] = \phi_{1}[\vec{\theta}]Op_{\Ic[\theta_s]}\phi_{2}[\vec{\theta}]$, $Op \in \{ \Uc, \Rc\}$, or propositional subformulas $\phi[\vec{\theta}] = p[\vec{\theta}]$. Then, given a timed state sequence $\tss = (\dsig,\sam)$, for $\vec{\theta}$, $\vec{\theta'}$ $\in \CoRe^n_{\ge 0}$, such that $\vec{\theta} \preceq \vec{\theta'}$, where $1 \leq j \leq n$, and for $i \in N$, we have:
\begin{itemize}
\item if for all such subformulas (i) $Op = \Uc$ and $\sup \Ic(\theta_s) = \theta_s$ or (ii) $Op = \Rc$ and $\inf \Ic(\theta_s) = \theta_s$ or (iii)  $p[\vec{\theta}] \equiv g(x) \leq \vec{\theta}$, then $\dle \phi[\vec{\theta}] \dri (\tss,i) \leq \dle \phi[\vec{\theta'}] \dri (\tss,i)$, i.e.,  function $\dle \phi[\vec{\theta}] \dri (\tss,i)$ is non-decreasing with respect to $\vec{\theta}$,
\item if for all such subformulas (i) $Op = \Rc$ and $\sup \Ic(\theta_s) = \theta_s$ or (ii) $Op = \Uc$ and $\inf \Ic(\theta_s) = \theta_s$ or (iii) $p[\vec{\theta}] \equiv g(x) \geq \vec{\theta}$, then $\dle \phi[\vec{\theta}] \dri (\tss,i) \geq \dle \phi[\vec{\theta'}] \dri (\tss,i)$, i.e.,  function $\dle \phi[\vec{\theta}] \dri (\tss,i)$ is non-increasing with respect to $\vec{\theta}$.
\end{itemize}
\end{theorem}

\begin{proof} [sketch]
The proof is by induction on the structure of the formula. The base case is given by Lemmas \ref{lem:timeOperMonDec} and \ref{lem:AtmPropMon}.

Consider the first case where $\phi[\vec{\theta}] = \phi_{1}[\vec{\theta}]\Un_{\Ic[\theta_s]}\phi_{2}[\vec{\theta}]$. Let $\vec{\theta}, \vec{\theta'} \in \CoRe_{\ge 0}^n$, where $\vec{\theta} \preceq \vec{\theta'}$. Let $i,j,k \in N$. Then $\Ic[\theta_{s}] \subseteq \Ic[\theta'_{s}]$ and, for all $j$, by the induction hypothesis we have
\begin{gather}
\min ( \dle \phi_2[\vec{\theta}] \dri (\tss,j) ) \le \min ( \dle \phi_2[\vec{\theta'}] \dri (\tss,j) ) \label{eq:thrm1} 
\end{gather}
For all k, by the induction hypothesis we have:
\begin{gather}
\inf\limits_{i\leq k <j} (\dle \phi_1[\vec{\theta}] \dri (\tss,k)) \bigr) \le \inf\limits_{i\leq k <j} ( \dle \phi_1[\vec{\theta'}] \dri (\tss,k)) \bigr) \label{eq:thrm2}
\end{gather}
Then by (\ref{eq:thrm1}) and (\ref{eq:thrm2}) we have
\begin{gather}
\dle \phi[\vec{\theta}]\dri(\mu,i) =\dle\phi_{1}[\vec{\theta}]\Un_{\Ic[\theta_{s}]}\phi_{2}[\vec{\theta}] \dri(\mu,i)=  \nonumber \\
\sup\limits_{j \in \sam^{-1}(\sam(i)+\Ic[\theta_{s}])} \bigl( \min ( \dle \phi_2[\vec{\theta}] \dri (\tss,j),  \inf\limits_{i\leq k <j} \dle \phi_1[\vec{\theta}] \dri (\tss,k)) \bigr) \nonumber \\
\leq \nonumber \\ \sup\limits_{j \in \sam^{-1}(\sam(i)+\Ic[\theta'_{s}])} \Bigl( \min ( \dle \phi_2[\vec{\theta'}] \dri (\tss,j), 
\inf\limits_{i\leq k <j} \dle \phi_1[\vec{\theta'}] \dri (\tss,k)) \Bigr) = \nonumber \\
\dle\phi_{1}[\vec{\theta'}]\Un_{\Ic[\theta'_{s}]}\phi_{2}[\vec{\theta'}] \dri(\mu,i) = 
\dle \phi[\vec{\theta'}]\dri(\mu,i) \nonumber
\end{gather}
Therefore, 
\begin{gather}
\dle \phi[\vec{\theta}]\dri(\mu,i) \le
\dle \phi[\vec{\theta'}]\dri(\mu,i) \tag*{\qed }
\end{gather}
\end{proof}

{\color{revComments}In this section, we have presented several cases where we can syntactically determine the monotonicity of the PMTL formula with respect to its parameters. \ifthenelse{\boolean{TECHREP}}{However, we remark that in general, determining the monotonicity of PMTL formulas is undecidable \cite{jin2013mining}.}{} }

%% file: parameter.tex
\section{Temporal Logic Parameter Bound Computation}
\label{sec:mtl:param:solution}

The notion of robustness of temporal logics will enable us to pose the parameter mining problem as an optimization problem.
In order to solve the resulting optimization problem, falsification methods and \staliro \cite{staliro:Online} can be utilized to estimate the solution for Problem \ref{prob:mtl:paramExploration}.

As described in the previous section, the parametric robustness functions that we are considering are monotonic with respect to the search parameters.
Therefore, if we are searching for a parameter vector over an interval $\Theta = [\vec{\thetaMin},\vec{\thetaMax}]$, where $\Theta$ is a hypercube and $\vec{\thetaMin} = [\thetaMin_1, \thetaMin_2, ..., \thetaMin_n]^\intercal$ and $\vec{\thetaMax} = [\thetaMax_1, \thetaMax_2, ..., \thetaMax_n]^\intercal$, we are either trying to minimize or maximize a function $f$ of $\vec{\theta}$ such that for all $\vec{\theta} \in \Theta^*$, we have $\dle \phi[\vec{\theta}] \dri (\Sys) \leq 0$.

\ifthenelse{\boolean{TECHREP}}{
\begin{exmp}[AT]
\label{exmp:param:rob:sys}
Let us consider again the automotive transmission example and the specification $\phi[\theta] = \Box_{[0,\theta]}p$ where $p \equiv (\omega \leq 4500)$.
The specification robustness $\dle \phi[\theta] \dri(\Delta_\Sys(u))$ as a function of $\theta$ and the input $u$ appears in Fig. \ref{fig:rob:2D:graph} for constant input signals.
The creation of the graph required $100 \times 30 = 3,000$ simulations.  
The contour under the surface indicates the zero level set of the robustness surface, i.e., the $\theta$ and $u$ values for which we get $\dle \phi[\theta] \dri(\Delta_\Sys(u)) = 0$.
From the graph, we can infer that $\theta^* \approx 2.8$ and that for any $\theta \in [2.8, 30]$, we have $\dle \phi[\theta] \dri (\Sys) \leq 0$.
The approximate value of $\theta^*$ is an estimate based on the granularity of the grid that we used to plot the surface.  $\hfill \blacktriangle$
\end{exmp}

\begin{figure}
\vspace{-10pt} 
\begin{center}
\includegraphics[width=6.2cm]{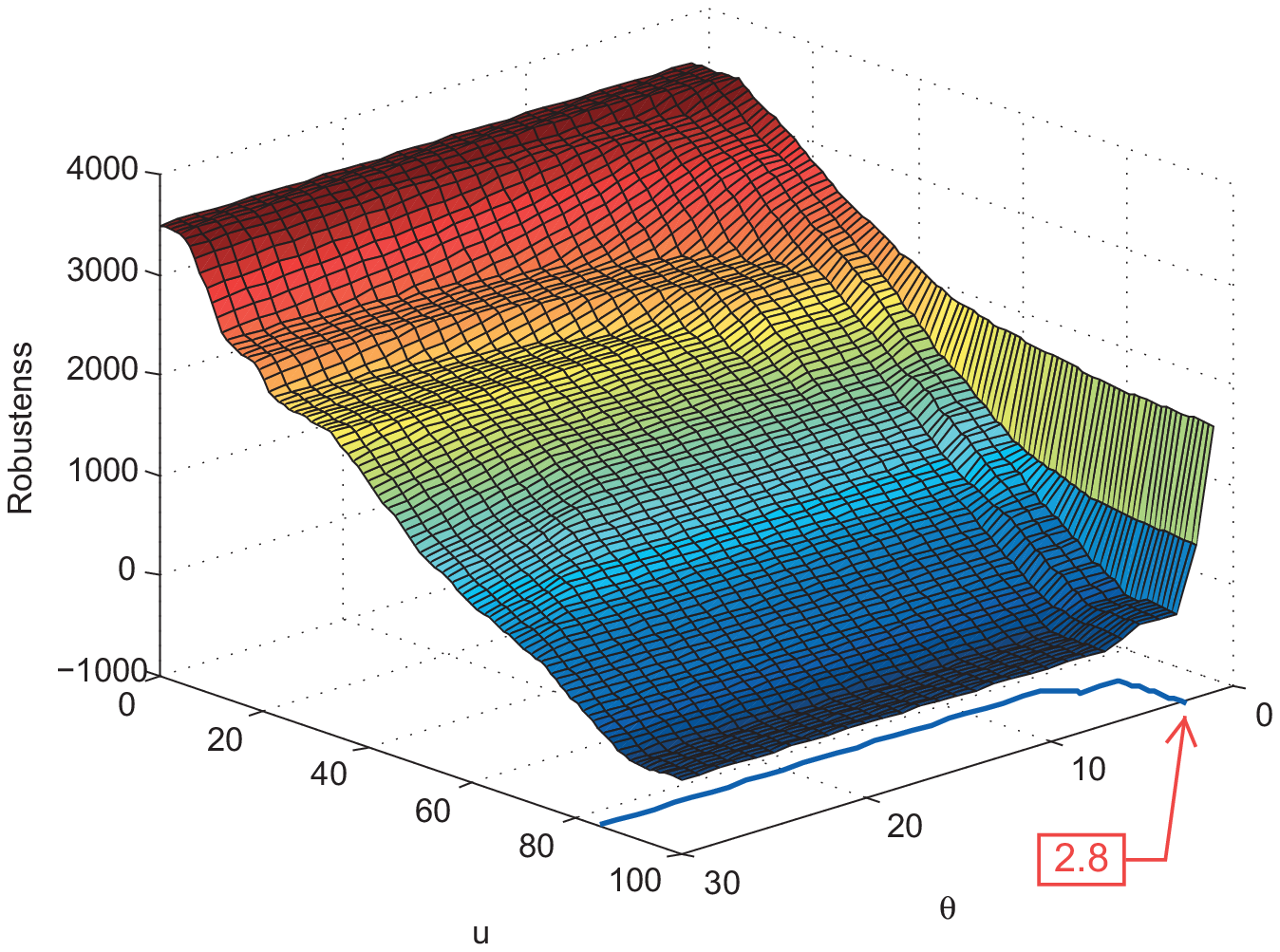}
\end{center}
\vspace{-10pt}
\caption{Example \ref{exmp:param:rob:sys}: Specification robustness estimate as a function of parameter $\theta$ and input $u$ for specification  $\phi[\theta] = \Box_{[0,\theta]}(\omega \leq 4500)$.}
\label{fig:rob:2D:graph}
\vspace{-10pt} 
\end{figure}
}{}
\ifthenelse{\boolean{TECHREP}}{
In summary, in order to solve Problem \ref{prob:mtl:paramExploration}, we would have to solve the following optimization problem:
}
{
In order to solve Problem \ref{prob:mtl:paramExploration}, we would have to solve the following optimization problem:	
}
\begin{align}
\label{eq:opt1}
\mbox{ optimize } \qquad & f(\vec{\theta}) \\
\mbox{ subject to } \qquad & \vec{\theta} \in \Theta \nonumber \mbox{ and } \\
 & \dle \phi[\vec{\theta}] \dri (\Sys) = \min_{\tss \in \Lc_{\sam}(\Sys)} \dle \phi[\vec{\theta}] \dri (\tss)\leq 0 \nonumber
\end{align}

Where $f:\mathbb{R}^n \rightarrow \mathbb{R}$ is a either a non-increasing ($\geq$) or a non-decreasing ($\leq$) function. For two vector parameter values $\vec{\theta}$, $\vec{\theta}'$, if $\vec{\theta} \preceq \vec{\theta}'$ and $\vec{\theta} \ge 0$ then $f(\vec{\theta}) \bowtie{} f(\vec{\theta}')$, where $\bowtie{} \in \{\geq,\leq \}$ depending on the monotonicity.

The function $\dle \phi[\vec{\theta}] \dri (\Sys)$  can not be computed using reachability analysis algorithms nor is known in closed form for the systems we are considering.
Therefore, we will have to compute an under-approximation of $\Theta^*$.
Our focus will be to formulate an optimization problem that can be solved using stochastic search methods. 
In particular, we will reformulate the optimization problem (\ref{eq:opt1}) into a new one where the constraints due to the specification are incorporated into the cost function:
\begin{equation}
\label{eq:opt2}
\mbox{optimize}_{\vec{\theta} \in \Theta} \left ( f(\vec{\theta}) + 
\left \{
\begin{array}{ll}
\gamma \pm \dle \phi[\vec{\theta}] \dri (\Sys) \\ \hspace{0.5cm} \mbox{ if } \dle \phi[\vec{\theta}] \dri (\Sys)  \geq 0 \\
0 \hspace{0.35cm} \mbox{ otherwise } 
\end{array}\right . \right )
\end{equation}
where the sign ($\pm$) and the parameter $\gamma$ depend on whether the problem is a maximization or a minimization problem.
The parameter $\gamma$ must be properly chosen so that the solution of problem (\ref{eq:opt2}) is in $\Theta$ if and only if $\dle \phi[\vec{\theta}] \dri (\Sys) \leq 0$. 
Therefore, if the problem in Eq. (\ref{eq:opt1}) is feasible, then the optimal points of equations (\ref{eq:opt1}) and (\ref{eq:opt2}) are the same.

 \ifthenelse{\boolean{TECHREP}}{\subsection{Non-increasing Robustness Functions}}{}

In the case of non-increasing robustness functions $\dle \phi [\vec{\theta}] \dri (\Sys)$ with respect to the search vector variable $\vec{\theta}$, the optimization problem is a minimization problem. Without loss of generality, let us consider the case for single parameter specifications. Assume that $\dle \phi[\thetaMax] \dri (\Sys) \leq 0$.
Since $\theta \leq \thetaMax$, we have $\dle \phi[\theta] \dri (\Sys) \geq \dle \phi[\thetaMax] \dri (\Sys)$, we need to find the minimum $\theta$ such that we still have $\dle \phi[\theta] \dri (\Sys) \leq 0$.
That $\theta$ value will be the desired $\theta^*$ since for all $\theta' \in [\theta^*,\thetaMax]$, we will have $\dle \phi[\theta'] \dri (\Sys) \leq 0$.

We will reformulate the problem of Eq. (\ref{eq:opt2}) so that we do not have to solve two separate optimization problems.
From (\ref{eq:opt2}), we have:
\begin{gather}
\min_{  \vec{\theta} \in \Theta} \left (f(\vec{\theta}) + 
\left \{
\begin{array}{ll}
\gamma + \min_{\tss \in \Lc_{\sam}(\Sys)} \dle \phi[\vec{\theta}] \dri (\tss) \\ \hspace{0.5cm} \mbox{ if } \min_{\tss \in \Lc_{\sam}(\Sys)} \dle \phi[\vec{\theta}] \dri (\tss)  \geq 0 \\
0 \hspace{0.35cm} \mbox{ otherwise } 
\end{array}\right . \right ) =
 \displaybreak[2] \nonumber \\
= \min_{  \vec{\theta} \in \Theta} \left (f(\vec{\theta}) + \min_{\tss \in \Lc_{\sam}(\Sys)}
\left \{
\begin{array}{ll}
 \gamma + \dle \phi[\vec{\theta}] \dri (\tss) \\ \hspace{0.5cm} \mbox{ if } \dle \phi[\vec{\theta}] \dri (\tss)  \geq 0 \\
0 \hspace{0.35cm} \mbox{ otherwise } 
\end{array}\right . \right ) =
 \displaybreak[2] \nonumber \\
= \min_{  \vec{\theta} \in \Theta} \min_{\tss \in \Lc_{\sam}(\Sys)} \left (f(\vec{\theta}) + 
\left \{
\begin{array}{ll}
\gamma + \dle \phi[\vec{\theta}] \dri (\tss) \\ \hspace{0.5cm} \mbox{ if } \dle \phi[\vec{\theta}] \dri (\tss) \geq 0 \\
0 \hspace{0.35cm} \mbox{ otherwise } 
\end{array}\right . \right ) 
\displaybreak[2]  \label{eq:opt3}
\end{gather} 
 
The previous discussion is formalized as follows.

\begin{figure}
\centering
\begin{tabular}{cc}
\includegraphics[width=3cm]{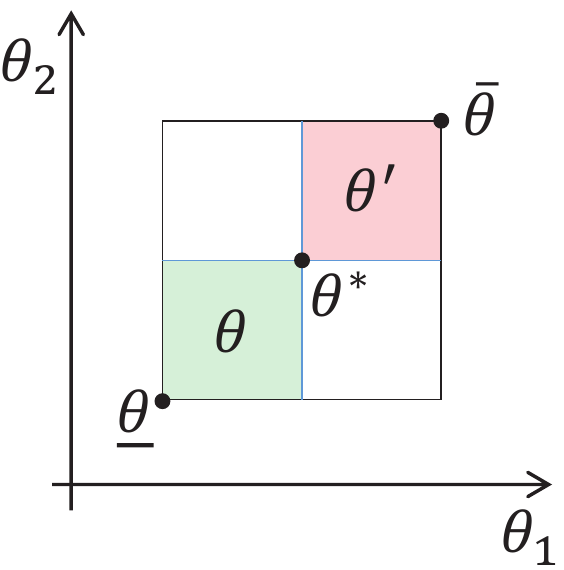} &
\includegraphics[width=3cm]{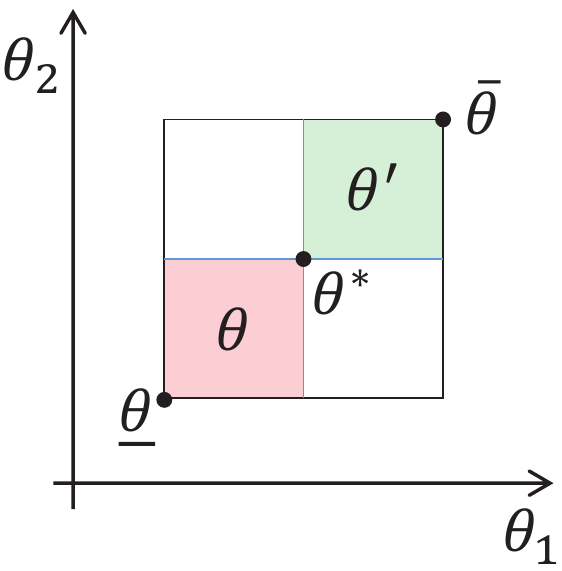}
\end{tabular}
\vspace{-5pt}
\caption{Illustration of the arrangement of parameters for non-increasing (\textbf{Left}) and non-decreasing (\textbf{Right}) robustness functions for a two parameter specification. The green (red) region represents parameter valuations for which we have a positive (negative) robustness value over all system behaviors.}
\label{fig:validGamma}
\vspace{-10pt} 
\end{figure}

\begin{prop}
Let $\vec{\theta}^*$ be a set of parameters and $\tss^*$ be the system trajectory returned by an optimization algorithm that is applied to the problem in Eq. (\ref{eq:opt3}).
If $\dle \phi[\vec{\theta}^*] \dri (\tss^*) \leq 0$, then for all $\vec{\theta} \succeq \vec{\theta}^*$, $\dle \phi[\vec{\theta}] \dri (\Sys) \leq 0$.
\end{prop}

\begin{proof}
If $\dle \phi[\vec{\theta}^*] \dri (\tss^*) \leq 0$, then $\dle \phi[\vec{\theta}^*] \dri (\Sys) \leq 0$.
Since $\dle \phi[\vec{\theta}] \dri (\Sys)$ is non-increasing with respect to $\vec{\theta}$, then for all $\vec{\theta} \succeq \vec{\theta}^*$, we also have $\dle \phi[\vec{\theta}] \dri (\Sys) \leq 0$.  $\hfill \qed$
\end{proof}

\begin{prop} If $f(\vec{\theta}) = \|\vec{\theta}\|$, and the robustness function is non-increasing, then $\gamma = \| \vec{\thetaMax} \| $ is a valid choice for parameter $\gamma$. Here, $\|\cdot\|$ denotes the euclidean norm.
\label{prop:nonincreasing}
\end{prop}
\begin{proof}
The interesting case to prove here is when we have $\vec{\theta}$ such that $  \dle \phi[\vec{\theta}] \dri(\Sigma) \ge 0 $  and we have $\vec{\theta}'$ such that $ \dle \phi[\vec{\theta}'] \dri(\Sigma) < 0$. See Fig. \ref{fig:validGamma} (Left) for an illustration of the arrangement of parameter valuations for a two parameter specification.

In this case
\begin{center}
$\gamma = \|\vec{\thetaMax} \| \ge \|\vec{\theta}'\| \ge \|\vec{\theta}\| $ \\ and \\  $\dle \phi[\vec{\theta}] \dri(\Sigma) \ge 0 \implies$ \\
$\|\vec{\theta}\| + \gamma + \dle \phi[\vec{\theta}] \dri(\Sigma) \ge \|\vec{\theta}'\|$
\end{center}
Therefore, if the problem in Eq. (\ref{eq:opt1}) is feasible, then the optimum of equations (\ref{eq:opt1}) and (\ref{eq:opt2}) is the same.$\hfill \qed$
\end{proof}



\ifthenelse{\boolean{TECHREP}}{
\begin{figure}
\vspace{-10pt} 
\centering
\includegraphics[width=6cm]{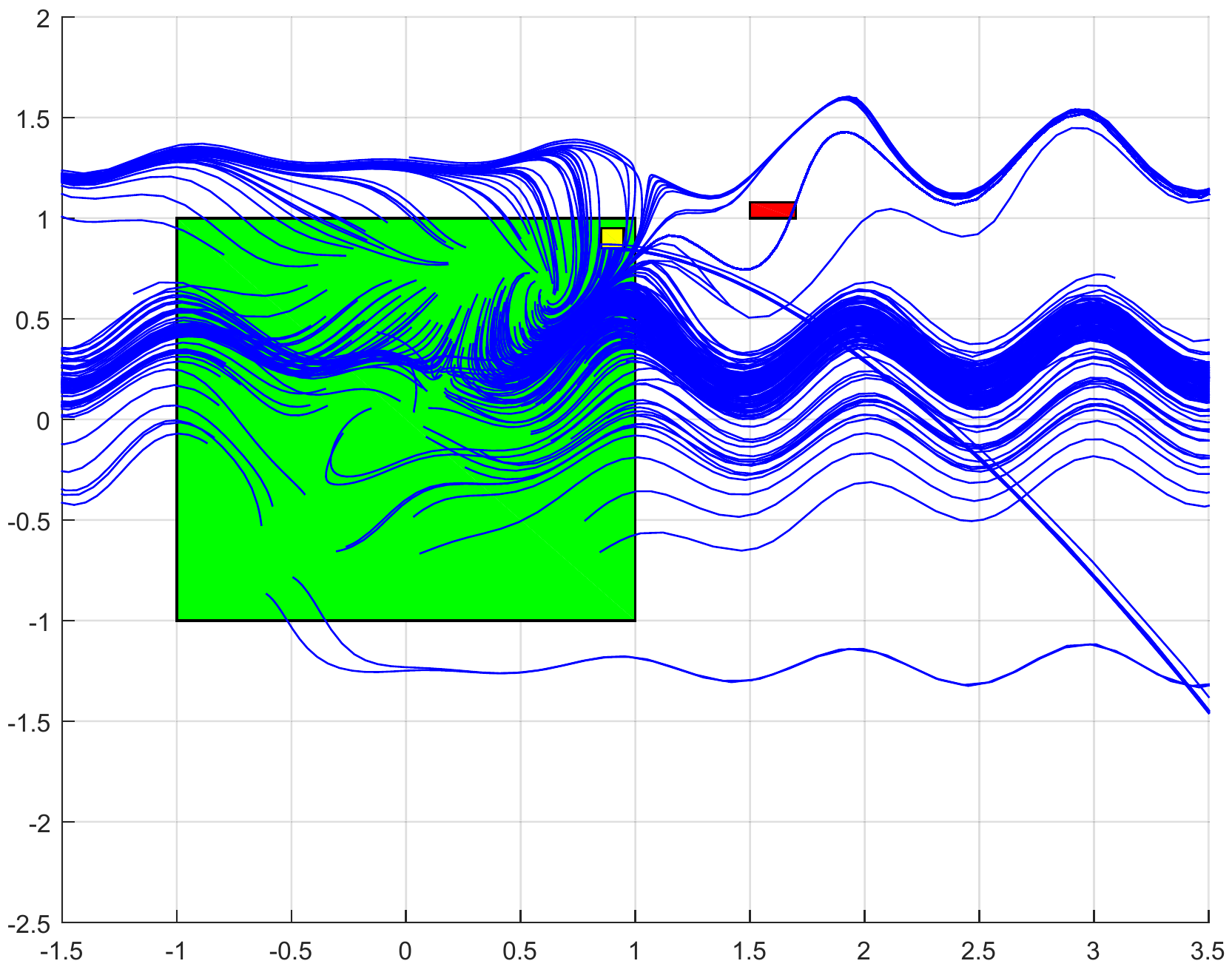} 
\vspace{-5pt}
\caption{Example \ref{exmp:hs}: Specification falsification for  $\phi[\vec{\theta}]$ = $\Box_{[0,\theta_1]} \neg a$ where $\Oc(a) = [1.5, \theta_2] \times [1, \theta_3]$ with mined parameters $\theta_1 = 3.417$, $\theta_2 = 1.7$, and $\theta_3 = 1.078$.}
\label{fig:hsex8}
\vspace{-15pt} 
\end{figure}
}{}

\ifthenelse{\boolean{TECHREP}}{
\begin{exmp}[AT]
Using Eq. (\ref{eq:opt3}) as a cost function, we can now compute a parameter for Example \ref{exmp:param:rob:sys} using \staliro \cite{AnnapureddyLFS11tacas,staliro:Online}. 
In particular, using Simulated Annealing as a stochastic optimization function, \staliro returns $\theta^* \approx 2.45$ as optimal parameter for constant input $u(t) = 99.81$.
The corresponding temporal logic robustness for the specification $\Box_{[0,2.45]}(\omega \leq 4500)$ is $-0.0445$.
The number of tests performed for this example was $500$ and, potentially, the accuracy of estimating $\theta^*$ can be improved if we increase the maximum number of tests.
However, based on 100 tests the algorithm converges to a good solution within $200$ tests.  $\hfill \blacktriangle$
\end{exmp}
}
{
\begin{exmp}[AT]
Let us consider again the automotive transmission example and the specification $\phi[\theta] = \Box_{[0,\theta]}p$ where $p \equiv (\omega \leq 4500)$.
Using Eq. (\ref{eq:opt3}) as a cost function, we can now compute a parameter using \staliro \cite{AnnapureddyLFS11tacas,staliro:Online}. 
In particular, using Simulated Annealing as a stochastic optimization function, \staliro returns $\theta^* \approx 2.45$ as optimal parameter for constant input $u(t) = 99.81$.
The corresponding temporal logic robustness for the specification $\Box_{[0,2.45]}(\omega \leq 4500)$ is $-0.0445$.
The number of tests performed for this example was $500$ and, potentially, the accuracy of estimating $\theta^*$ can be improved if we increase the maximum number of tests.
However, based on 100 runs the algorithm converges to a good solution within $200$ tests.  $\hfill \blacktriangle$	
\end{exmp}
}

\ifthenelse{\boolean{TECHREP}}{
\begin{figure*}
\centering
\begin{tabular}{cc}
\includegraphics[width=6.2cm]{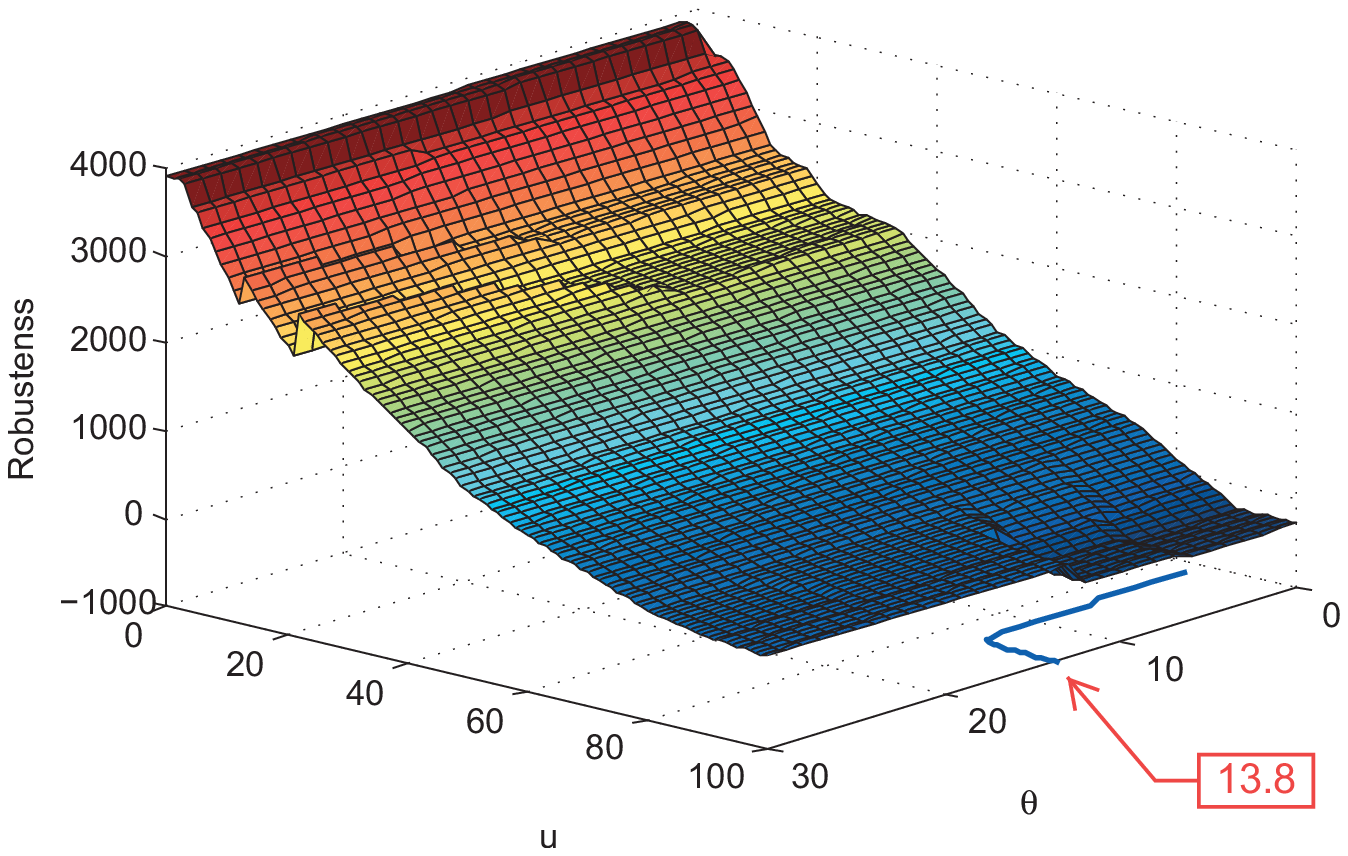} &
\includegraphics[width=6.1cm]{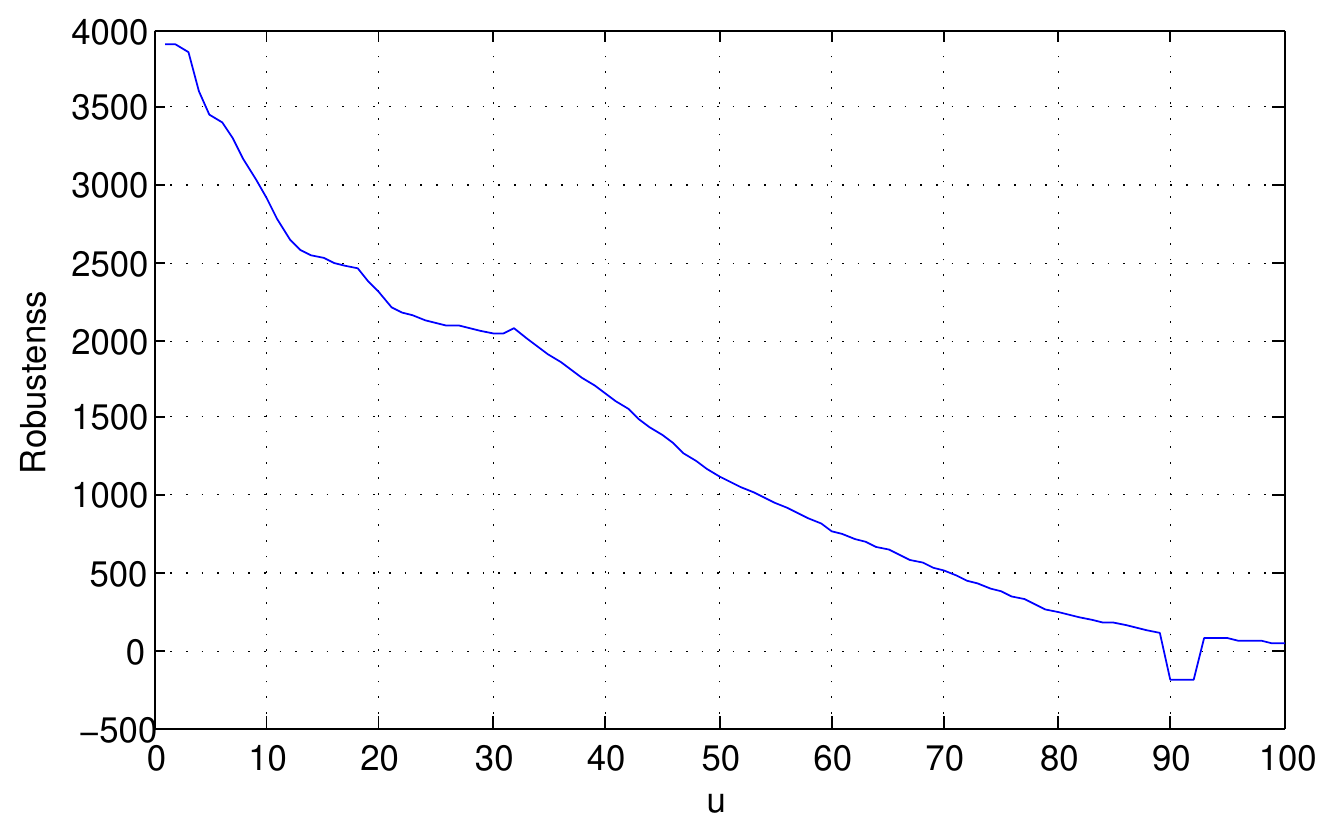} 
\end{tabular}
\caption{Example \ref{exmp:param:rob:sys:2}. \textbf{Left}: Specification robustness as a function of the parameter $\theta$ and the input $u$. \textbf{Right}: The robustness function $\dle \Box_{[12.59,30]}(\omega \leq 4500) \dri (\Delta_{\Sys}(u))$.}
\label{fig:rob:2D:graph:2}
\vspace{-10px}
\end{figure*}
}{}

\ifthenelse{\boolean{TECHREP}}{
\begin{exmp}[HS]
\label{exmp:param:rob:sys:HS}
Let us consider the specification $\phi[\vec{\theta}]$ = $\Box_{[0,\theta_1]} \neg a$ where $\Oc(a) = [1.5, \theta_2] \times [1, \theta_3]$ on our hybrid system running example. Here, the bounds for the timing parameter are $\theta_1 \in [0, 5]$ and the bounds for the state parameters are $\theta_2 \in [1.5, 2.1]$ and $\theta_3 \in [1.1, 1.6]$. The ranges for the parameters are chosen based on prior knowledge and experience about the system. The parameter mining algorithm from \staliro returns  $\theta^*_1 = 3.417$, $\theta^*_2 = 1.7$, and $\theta^*_3 = 1.078$ after running 1000 tests on the system. The generated trajectories by the parameter mining algorithm are presented in Fig. \ref{fig:hsex8}. The returned parameters guarantee that the system does not satisfy the specification for all parameters $\vec{\theta}$ where $\vec{\theta}^* \preceq \vec{\theta}$.  $\hfill \blacktriangle$
\end{exmp}
}{}

\ifthenelse{\boolean{TECHREP}}{}{We note that the parameter bound computation for non-decreasing robustness functions is symmetric to the solution proposed in this section and, therefore, omitted. The interested reader is referred to the extended version of the paper in \cite{Hoxha2016MiningExtendedVer}. }

\ifthenelse{\boolean{TECHREP}}{
\subsection{Non-decreasing Robustness Functions}

The case of non-decreasing robustness functions is symmetric to the case of non-increasing robustness functions.
In particular, the optimization problem is a maximization problem.
We will reformulate the problem of Eq. (\ref{eq:opt2}) so that we do not have to solve two separate optimization problems.
From (\ref{eq:opt2}), we have:
\begin{gather}
\max_{  \vec{\theta} \in \Theta} \left (f(\vec{\theta}) +
\left \{
\begin{array}{ll}
\gamma - \max_{\tss \in \Lc_{\sam}(\Sys)} \dle \phi[\vec{\theta}] \dri (\tss) \\ \hspace{0.5cm} \mbox{ if } \max_{\tss \in \Lc_{\sam}(\Sys)} \dle \phi[\vec{\theta}] \dri (\tss)  \geq 0 \\
0 \hspace{0.35cm} \mbox{ otherwise } 
\end{array}\right . \right ) =
\displaybreak[2] \nonumber \\
= \max_{  \vec{\theta} \in \Theta} \left (f(\vec{\theta}) + \max_{\tss \in \Lc_{\sam}(\Sys)}
\left \{
\begin{array}{ll}
\gamma - \dle \phi[\vec{\theta}] \dri (\tss) \\ \hspace{0.5cm} \mbox{ if }  -\dle \phi[\vec{\theta}] \dri (\tss)  \leq 0 \\
0 \hspace{0.35cm} \mbox{ otherwise } 
\end{array}\right . \right ) 
= \displaybreak[2] \nonumber \\
= \max_{  \vec{\theta} \in \Theta} \max_{\tss \in \Lc_{\sam}(\Sys)} \left (f(\vec{\theta}) + 
\left \{
\begin{array}{ll}
\gamma - \dle \phi[\vec{\theta}] \dri (\tss) \\ \hspace{0.5cm} \mbox{ if } \dle \phi[\vec{\theta}] \dri (\tss) \geq 0 \\
0 \hspace{0.35cm} \mbox{ otherwise } 
\end{array}\right . \right ) 
\displaybreak[2]  \label{eq:opt4}
\end{gather}

The previous discussion is formalized in the following result.


\begin{prop}
Let $\vec{\theta}^*$ be a set of parameters and $\tss^*$ be the system trajectory returned by an optimization algorithm that is applied to the problem in Eq. (\ref{eq:opt4}).
If $\dle \phi[\vec{\theta}^*] (\tss^*) \leq 0$, then for all $\vec{\theta} \preceq \vec{\theta}^*$, we have $\dle \phi[\theta] \dri (\Sys) \leq 0$.
\end{prop}

\begin{proof}
If $\dle \phi[\vec{\theta}^*] \dri (\tss^*) \leq 0$, then $\dle \phi[\vec{\theta}^*] \dri (\Sys) \leq 0$.
Since $\dle \phi[\vec{\theta}] \dri (\Sys)$ is non-decreasing with respect to $\vec{\theta}$, then for all $\vec{\theta} \preceq \vec{\theta}^*$, we also have $\dle \phi[\vec{\theta}] \dri (\Sys) \leq 0$. $\hfill \qed$
\end{proof}

\begin{prop} If $f(\vec{\theta}) = \|\vec{\theta}\|$ and the robustness function is non-decreasing, then $\gamma = - \| \vec{\thetaMax} \| $ is a valid choice for parameter $\gamma$.
\label{prop:nondecreasing}
\end{prop}
\begin{proof}
The interesting case to prove here is when we have $\vec{\theta}$ such that $  \dle \phi[\vec{\theta}] \dri(\Sigma) < 0 $  and we have $\vec{\theta}'$ such that $ \dle \phi[\vec{\theta}'] \dri(\Sigma) \ge 0$. See Fig. \ref{fig:validGamma} (Right) for an illustration of the arrangement of parameter valuations for a two parameter specification. In this case 
\begin{center}
 $\gamma = - \|\vec{\thetaMax} \|$, $\dle \phi[\vec{\theta}'] \dri(\Sigma) \ge 0$ and \\ $\|\vec{\thetaMax} \| \ge \|\vec{\theta}'\| \ge \|\vec{\theta}\| \implies$ \\
$\|\vec{\theta}\| \ge \|\vec{\theta}'\| + \gamma - \dle \phi[\vec{\theta}'] \dri(\Sigma)$
\end{center}
Therefore, if the problem in Eq. (\ref{eq:opt1}) is feasible, then the optimum of equations (\ref{eq:opt1}) and (\ref{eq:opt2}) is the same. $\hfill \qed$
\end{proof}


\begin{exmp}[AT]
\label{exmp:param:rob:sys:2}
Let us consider the specification $\phi[\theta]$ $=$ $\Box_{[\theta,30]}$ $(\omega \leq 4500)$ on our running example.
The specification robustness $\dle \phi[\theta] \dri(\Delta_\Sys(u))$ as a function of $\theta$ and the input $u$ appears in Fig. \ref{fig:rob:2D:graph:2} for constant input signals.
The creation of the graph required $100 \times 30 = 3,000$ tests.  
The contour under the surface indicates the zero level set of the robustness surface, i.e., the $\theta$ and $u$ values for which we get $\dle \phi[\theta] \dri(\Delta_\Sys(u)) = 0$.
We remark that the contour is actually an approximation of the zero level set computed by a linear interpolation using the neighboring points on the grid.
From the graph, we could infer that $\theta^* \approx 13.8$ and that for any $\theta \in [0, 13.8]$, we would have $\dle \phi[\theta] \dri (\Sys) \leq 0$.
Again, the approximate value of $\theta^*$ is a rough estimate based on the granularity of the grid.

Using Eq. (\ref{eq:opt4}) as a cost function, we can now compute a parameter for Example \ref{exmp:param:rob:sys:2} using our toolbox \staliro \cite{AnnapureddyLFS11tacas,staliro:Online}. 
\staliro returns $\theta^* \approx 12.59$ as optimal parameter for constant input $u(t) = 90.88$ within 250 tests.
The temporal logic robustness for the specification $\Box_{[12.59,30]}(\omega \leq 4500)$ with respect to the input $u$ appears in Fig. \ref{fig:rob:2D:graph:2} (Right). $\hfill \blacktriangle$
%
\end{exmp}

}{}

%% file: parameterDomain.tex

\section{Parameter Falsification Domain}
\label{sec:mtl:paramFalsDomain}

\begin{figure*}
\centering
\begin{tabular}{cccc}
\includegraphics[width=4cm]{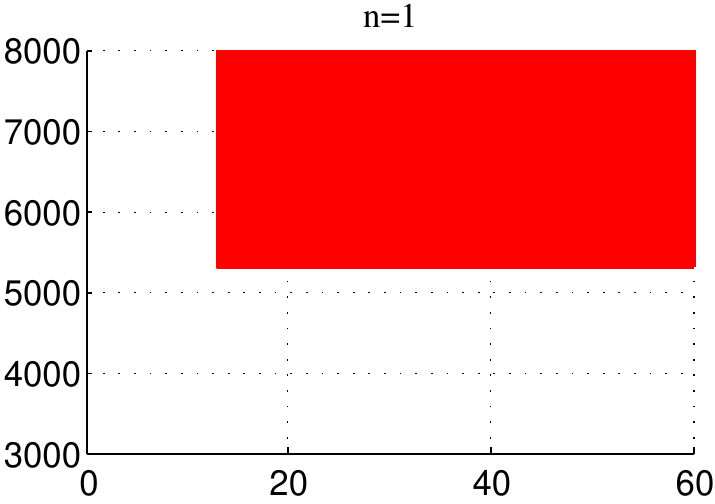} &
\includegraphics[width=4cm]{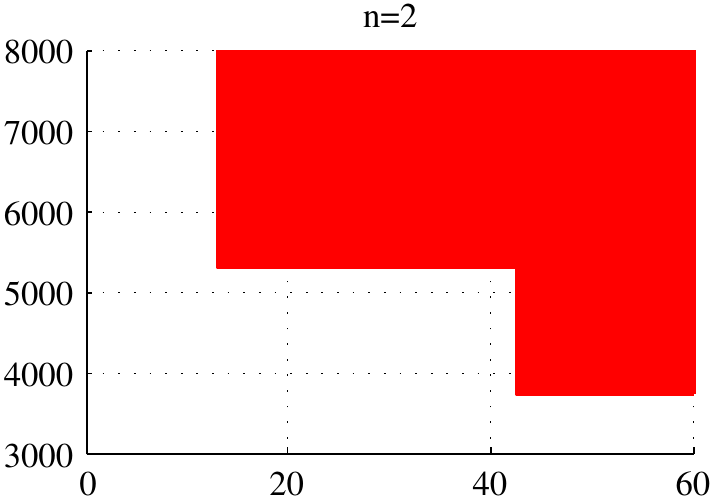} & 
\includegraphics[width=4cm]{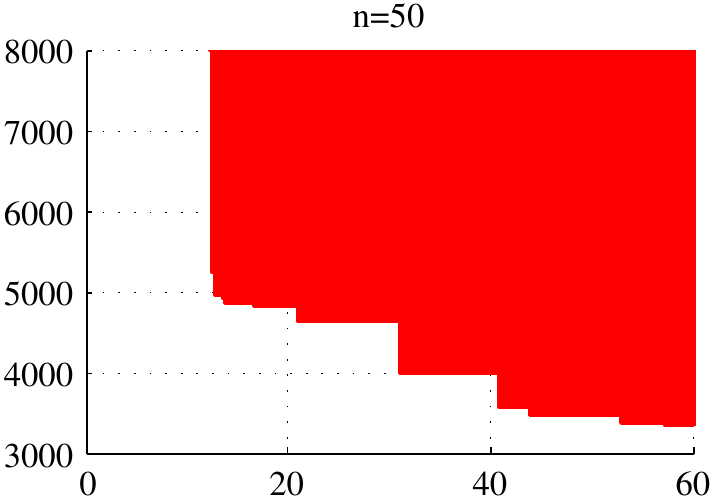} &
\includegraphics[width=4cm]{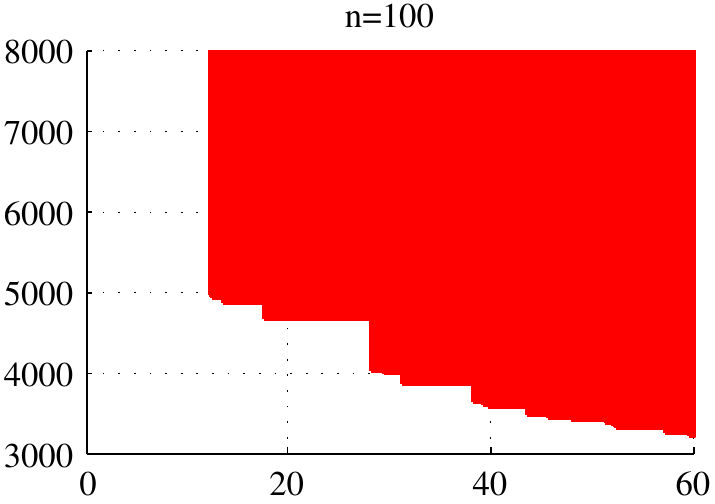} 
\end{tabular}
\vspace{-4pt}
\caption{ Illustration of the iterative process for Algorithm \ref{alg:RGPFDA}. Specification: $\phi[\vec{\theta}] = \neg (\Diamond_{[0,\theta_1]} q \wedge \Box p[\theta_2] )$  where $p[\theta_2] \equiv (\omega \leq \theta_2)$ and $q \equiv (v \geq 100)$. Model: Automatic Transmission as described in Example \ref{exmp:autotrans}. The red colored set represents set $\Psi = \{\vec{\theta} \in \Theta \ | \ \Sys \not \models \phi[\vec{\theta}] \}$ i.e. the set of parameter values such that the system does not satisfy the specification. {\color{revComments} In each iteration of the algorithm, set $\Psi$ gets expanded by the optimal falsifying parameter which is guided by the robustness landscape and the random weight in the priority function.} }
\label{fig:rob:exmp:validityDomain}
\vspace{-12pt}
\end{figure*}

We utilize the solution of Problem \ref{prob:mtl:paramExploration} and exploit the robustness landscape of a specific class of temporal logic formulas to present two algorithms to estimate $\Psi = \{\vec{\theta}^* \in \Theta \ | \ \Sys \not \models \phi[\vec{\theta}^*] \}$ for Problem \ref{prob:mtl:paramExploration}. 
In fact, we can reduce this problem to finding the set $\Theta^{bd} = \Psi \cap \{\vec{\theta}^* \in \Theta \ | \ \dle \phi[\vec{\theta}^*] \dri (\Sys) = 0\}$ since the robustness landscape is monotonic. 
Here, $\Theta^{bd}$ represents the intersection of the robustness function with the zero level set. 
As a preprocessing step, the PMTL parameters are normalized in the range $[0,1]$ to avoid bias during the optimization process. {\color{revComments} It is important to note, that due to the undecidable nature of the problem, we cannot determine satisfying parameter values. Therefore, we generate the parameter falsification domain by finding only falsifying parameter values.}


\ifthenelse{\boolean{TECHREP}}{The first method approximates  $\Theta^{bd}$ by modifying the priority function $f$ and thereby slightly shifting the minimum or maximum of the objective function in Eq. \ref{eq:opt3} or Eq. \ref{eq:opt4}, respectively.}{The first method approximates  $\Theta^{bd}$ by modifying the priority function $f$ and thereby slightly shifting the minimum or maximum of the objective function in Eq. \ref{eq:opt3}.}  The magnitude of the shift depends on the shape of the robustness landscape of the model and specification.

As shown in Algorithm \ref{alg:RGPFDA}, the set $\Psi$ is explored iteratively. 
For every iteration, we draw a random vector $\omega$ with dimension equal to the dimension of $\Theta$. 
The random vector is used as parameter weights for the priority function $f(\vec{\theta})$. 
Namely, $f(\vec{\theta}) = \sum_{}^{} w_i\theta_i$.
We run parameter mining, which returns an approximation for Eq. (\ref{eq:opt2}). 
In case $\phi[\vec{\theta}]$ is non-decreasing (or non-increasing), the optimization algorithm $\mathtt{opt}$ is a maximization (or minimization) algorithm. 
We utilize the values mined and the corresponding robustness value to expand $\Psi$ and reduce the unknown parameter range for the next iteration. 
We present the iterative process in Fig.~\ref{fig:rob:exmp:validityDomain}.

We define a PMTL specification monotonicity function $\mathcal{M} : $ PMTL $ \rightarrow \{-1,0,1\}$ where
\[ \mathcal{M}(\phi[\vec{\theta}]) = \left\{ \begin{array}{ll}
         1 & \mbox{if $\phi[\vec{\theta}]$ is non-decreasing};\\
         -1 & \mbox{if $\phi[\vec{\theta}]$ is non-increasing};\\
         0 & \mbox{otherwise}.\end{array} \right. \] 

A monotonicity computation algorithm is presented in \cite{AsarinDMN12rv} and generalized in \cite{jin2013mining}.

\begin{algorithm} 
  \caption{ Robustness Guided Parameter Falsification Domain Algorithm RGDA($\mathtt{opt}$, $\Gamma$, $\Theta$, $\phi$, $\Sigma$, $n$, $t$)}
  {\color{revComments} {\bf Input}: Stochastic optimization algorithm $\mathtt{opt}$, search space $\Gamma$, parameter range $\Theta$, specification $\phi$, system $\Sigma$, number of iterations $n$ and tests $t$ \\
  {\bf Output}: Parameter falsification domain $\Psi$ \\
  {\bf Internal Variables}: Parameter weights $\vec{\omega}$, parameters mined $\vec{\theta}^*$ and robustness value $\gamma$ }
  \label{alg:RGPFDA}
  \begin{algorithmic}[1]
	\State $\langle \Psi$, $\vec{\omega}$, $\vec{\theta}^*$, $\gamma \rangle \gets \langle \emptyset$, $\emptyset$, $\emptyset$, $\emptyset \rangle$
	\For {$i=0$ to $n$}
		\State $\vec{\omega} \gets $ \textsc{RandomVector}$([0,1],$ \textsc{dimension}$(\Theta))$ 
		\State $[\vec{\theta}^*, \gamma] \gets \mathtt{opt}(\Gamma,\Theta,\phi,\Sigma,t,\vec{\omega}, \mathcal{M}(\phi[\vec{\theta}^*]))$
		\Comment{run parameter mining and robustness computation}
		
		\If{ $(\gamma \le 0$) }
					\If{ ($\mathcal{M}(\phi[\vec{\theta}^*])=1$)}
						\State $\Psi \gets \Psi \cup \{\vec{\theta} \in \Theta \ | \ \forall i \  (0 \le \theta_i \le \theta^*_i)\}$ \Comment{expand the falsification domain $\Psi$}
					\ElsIf{ ($\mathcal{M}(\phi[\vec{\theta}^*])=-1$)} 
						\State $\Psi \gets \Psi \cup \{\vec{\theta} \in \Theta \ | \ \forall i \  (\theta_i \ge \theta^*_i \ge 0)\}$
					\EndIf
		\EndIf
	\EndFor
	
	\State \Return $\Psi$
  \end{algorithmic}
\end{algorithm}

Algorithm \ref{alg:newalg} explores the set $\Theta^{bd}$ by iteratively expanding the set of falsifying parameters, namely, the set $\Psi$. However in this case, the search is finely structured and does not depend on randomized weights.
For presentation purposes, let us consider the case for specifications with non-decreasing monotonicity. 
Given a normalized parameter range with dimension $\eta$, in each iteration of the algorithm, we solve the following optimization problem: 
\begin{align}
\label{eq:optAlg}
\mbox{ maximize } \qquad & c \\
\mbox{ subject to } \qquad & \ c * \vec{b} + \vec{p} \in \Theta \nonumber \mbox{ and } \\
 & \Sys \not \models \phi[ \ c * \vec{b} + \vec{p} \ ]  \nonumber
\end{align}

\noindent where $\vec{p}$ is the starting point of the optimization problem in each iteration and $\vec{b}$ is the bias vector which enables to prioritize specific parameters in the search. 
Namely, the choice of $\vec{b}$ directs the expansion of the parameter falsification domain along a specific direction. 
We refer to the solution of Eq. \ref{eq:optAlg} in the $i^{th}$ iteration of the algorithm as \textbf{marker($i$)}.
Initially, for the first iteration, the value of $\vec{p}$ is set to $\vec{0}$ or $\vec{1}$ depending on the monotonicity of the specification. 
The returned \textbf{marker($1$)} from Eq. \ref{eq:optAlg} is then utilized to update $\Psi$, the set of parameters for which the system does not satisfy the specification.
Next, we generate at most $2^\eta - 2$ initial position vectors induced by the returned \textbf{marker($1$)}. 

\begin{figure*}
\centering
\begin{tabular}{c}
\includegraphics[width=17.5cm]{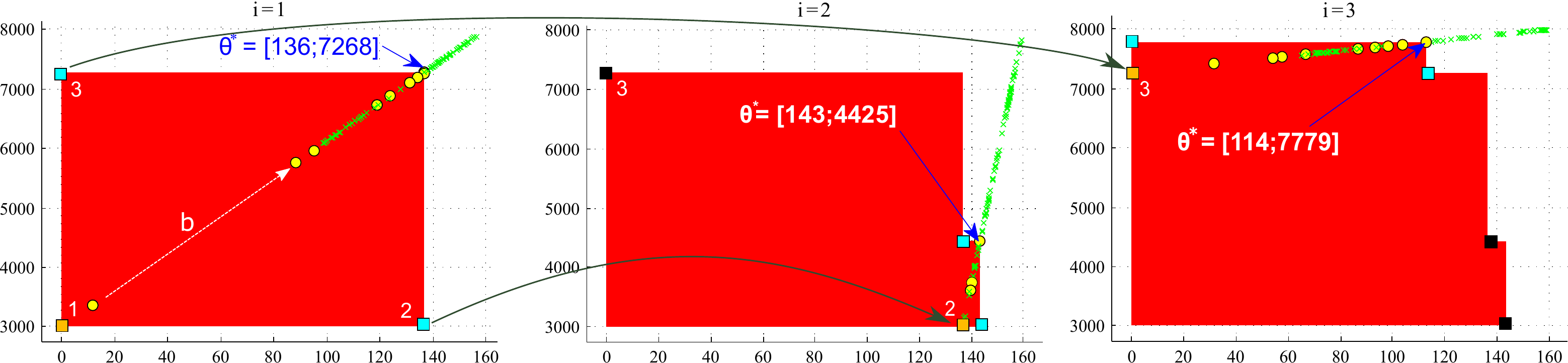} 
\end{tabular}
\vspace{-4pt}
\caption{ Illustration of the iterative process for Algorithm \ref{alg:newalg}. Specification: $\phi[\vec{\theta}] = \Box ( p[\theta_1] \wedge q[\theta_2] )$  where $p[\theta_1] \equiv (v \leq \theta_1)$ and $q \equiv (\omega \leq \theta_2)$. Model: Automatic Transmission as described in Example \ref{exmp:autotrans}. {\color{revComments} The parameter range for the specification is $\Theta = [0 \ 160; 3000 \ 8000]$. 
In each plot, the search is conducted in a specific direction \vec{b}.
The plots from left to right represent three iterations of Algorithm \ref{alg:newalg}. 
The yellow circles and green marks represent sample points of the search optimizer in the process of solving Eq. \ref{eq:optAlg}.
Specifically, the yellow circles represent parameter values for which we have found system inputs and initial conditions that falsify the specification. 
The green marks represent parameter values for which falsification is not found. 
The largest yellow circle found by the stochastic optimizer is returned as the current marker.
The orange squares represent the initial position of the search in the current iteration.
The blue squares represent the initial positions generated by the current marker that will be considered in future iterations. 
The black squares represent initial positions that will be considered in future iterations.
The red colored set represents set $\Psi = \{\vec{\theta} \in \Theta \ | \ \Sys \not \models \phi[\vec{\theta}] \}$ i.e. the set of parameter values such that the system does not satisfy the specification.} }
\label{fig:rob:exmp:validityDomainSDP}
\vspace{-12pt}
\end{figure*}

{\color{revComments} Consider the example presented in Fig. \ref{fig:rob:exmp:validityDomainSDP} where we have \textbf{marker($1$)} $= [136;7268]$. 
That value is utilized to update $\Psi$ and generate two new initial position vectors at $[0; 7268]$ and $[136;0]$.
In the next iteration of the algorithm, the search is initialized in one of the newly generated initial position vectors.
Namely, the search starts in $[0; 7268]$ or $[136;0]$ (see Fig. \ref{fig:rob:exmp:validityDomainSDP}, Left). The initial position vector not utilized is stored in a list and used in future iterations. In the second iteration, $[136;0]$ is used as the initial position vector. We return the solution to Eq. \ref{eq:optAlg} with \textbf{marker($2$)} $= [143;4425]$ which generates the initial position vectors $[143;0]$ and $[136;4425]$ (Fig. \ref{fig:rob:exmp:validityDomainSDP}, Middle). Similarly, \textbf{marker($3$)} is generated in Fig. \ref{fig:rob:exmp:validityDomainSDP} (Right). 
In this example, the directional vector \vec{b}, in each iteration, directs towards the bounds of the parameter range, namely $(160, 8000)$.} 
The algorithm terminates when one of the following conditions is met: 1) The distance between \textbf{markers} is less than some value $\epsilon$, or 2) no new markers are generated from the current set of initial position vectors, or 3) {\color{revComments} a maximum number of iterations is exceeded}. 


\begin{algorithm} 
  \caption{Structured Parameter Falsification Domain Algorithm SDA($\mathtt{opt}$, $\Gamma$, $\Theta$, $\phi$, $\Sigma$, $t$, $\epsilon$, $\vec{b}$, {\color{revComments} $n$})}
  {\color{revComments} {\bf Input}: Stochastic optimization algorithm $\mathtt{opt}$, search space $\Gamma$, parameter range $\Theta$, specification $\phi$, system $\Sigma$, number of tests $t$, minimum distance between markers $\epsilon$, bias vector $\vec{b}$, maximum number of iterations $n$ \\
  {\bf Output}: Parameter falsification domain $\Psi$ \\
  {\bf Internal Variables}: List of initial positions $\mathcal{ML}$, termination condition $\mathcal{TC}$, initial positions generated in the current iteration $\mathcal{TL}$, iteration $i$ }
  \label{alg:newalg}
  \begin{algorithmic}[1]
    \State {\color{revComments} $\langle$$\Psi$, $\vec{p}$, $\mathcal{TC}$, $\mathcal{ML}$, $\mathcal{TL}$, $i$$\rangle$ $\gets$  $\langle \emptyset$, $\emptyset$, $\bot$, $\{\}$, $\{\}$, $0$$\rangle$  }
	
	\If{ ($\mathcal{M}(\phi[\vec{\theta}])=1$)}
		\State $\mathcal{ML}$.\textsc{Add}($\vec{0}($\textsc{dimension}$(\Theta))$)
    \ElsIf{ ($\mathcal{M}(\phi[\vec{\theta}])=-1$)} 	
    	\State $\mathcal{ML}$.\textsc{Add}($\vec{1}($\textsc{dimension}$(\Theta))$)
	\EndIf
	
	
	\While {$\mathcal{TC} = \bot$}
		\State $\mathcal{TL} \gets \{\}$ 
		\For {$\vec{v}$ in $\mathcal{ML}$}
			\State {\color{revComments} $i \gets i+1$}
			\State $[\vec{\theta}^*, \gamma] \gets \mathtt{opt}(\Gamma,\Theta,\phi,\Sigma,t,\omega, \mathcal{M}(\phi[\vec{\theta}]), \vec{b}, \vec{v})$ 		\Comment{run parameter mining starting at $\vec{v}$ and search along the directional vector $\vec{b}$}

				\If{ $( \gamma \le 0$) }
							\State $\mathcal{TL}$.\textsc{Add}(\textsc{GenerateMarkers}($\theta^*$, $\mathcal{M}(\phi[\vec{\theta}])$))		
							\If{ ($\mathcal{M}(\phi[\vec{\theta}^*])=1$)}
								\State $\Psi \gets \Psi \cup \{\vec{\theta} \in \Theta \ | \ \forall i \  (0 \le \theta_i \le \theta^*_i\}$
								\State $\Theta \gets \Theta \setminus \Psi$
							\ElsIf{ ($\mathcal{M}(\phi[\vec{\theta}^*])=-1$)} 
								\State $\Psi \gets \Psi \cup \{\vec{\theta} \in \Theta \ | \ \forall i \  (\theta_i \ge \theta^*_i \ge 0)\}$
								\State $\Theta \gets \Theta \setminus \Psi$
							\EndIf
					\EndIf

		\EndFor
		
		\State $\mathcal{ML} \gets \mathcal{TL}$
		
		\If {$\mathcal{ML}$.\textsc{IsEmpty}() or \textsc{DistanceBetweenMarkers}($\mathcal{ML}$) $< \epsilon$ or {\color{revComments} $i > n$}}
			$\mathcal{TC} \gets \top$ 
		\EndIf	
	\EndWhile	
	\State \Return $\Psi$
  \end{algorithmic}
\end{algorithm}
\vspace{-10pt}

%% file: examples.tex

\section{Experiments and a Case Study}

The algorithms and examples presented in this work are implemented and publicly available through the Matlab toolbox \staliro{} \cite{AnnapureddyLFS11tacas,staliro:Online}.

The parametric MTL exploration of CPS is motivated by a challenge problem published by Ford in 2002 \cite{ChutinanB02fordtech}.
In particular, the report provided a simple -- but still realistic -- model of a powertrain system (both the physical system and the embedded control logic) and posed the question whether there are constant operating conditions that can cause a transition from gear two to gear one and then back to gear two.
That behavior would imply that the gear transition from 1 to 2 was not necessary in the first place.

\ifthenelse{\boolean{TECHREP}}{
The system is modeled in Checkmate \cite{SilvaK00acc}.
It has 6 continuous state variables and 2 Stateflow charts with 4 and 6 states, respectively.
The Stateflow chart for the shift scheduler appears in Fig. \ref{fig:powertrain:graphs}. 
The system dynamics and switching conditions are linear.
However, some switching conditions depend on the initial conditions of the system.
The latter makes the application of standard system verification tools not a straightforward task.

\begin{figure*}
\centering
\begin{tabular}{cc}
\multirow{2}{*}[55pt]{\includegraphics[width=6.1cm]{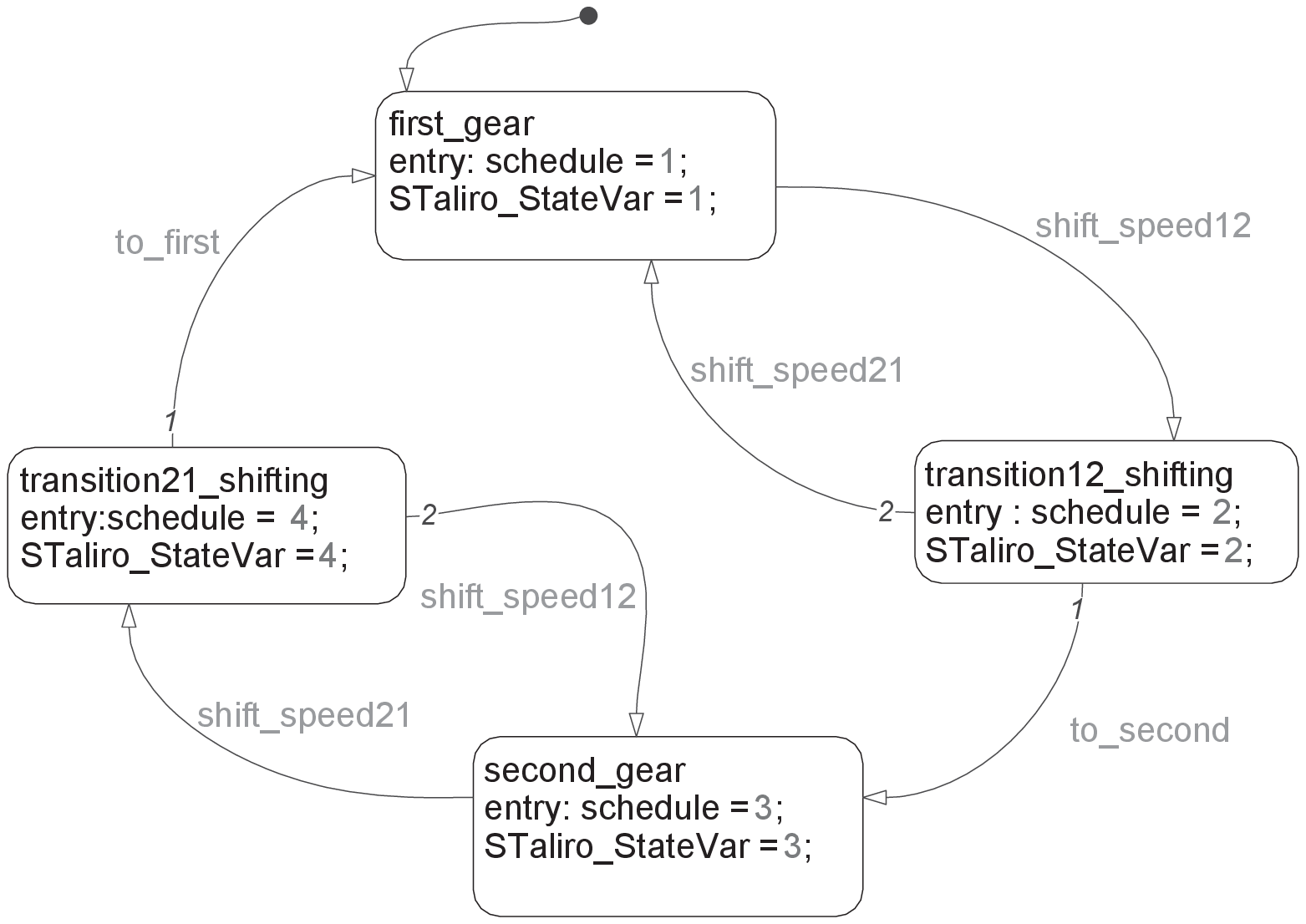}} &
\includegraphics[width=6.1cm]{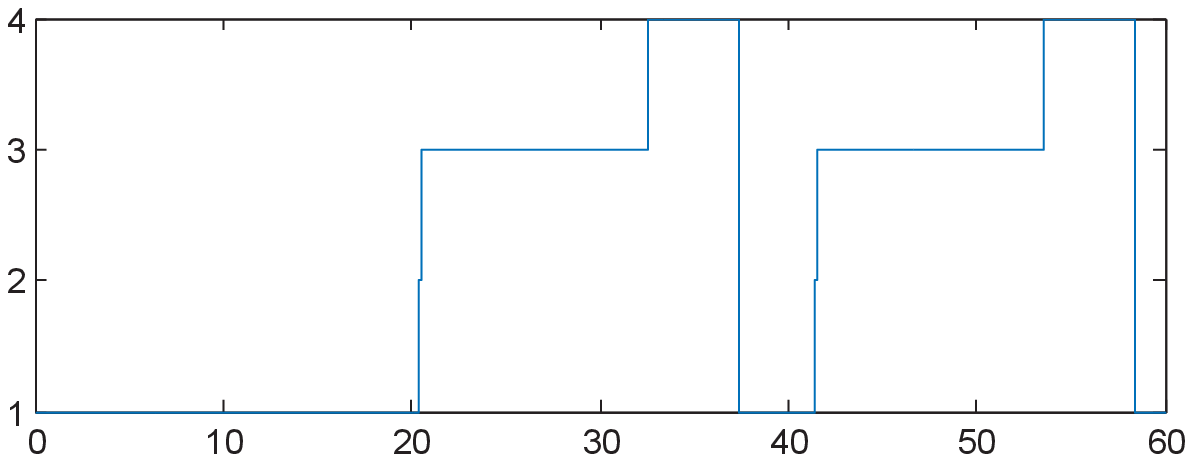} \\
& \includegraphics[width=5.9cm]{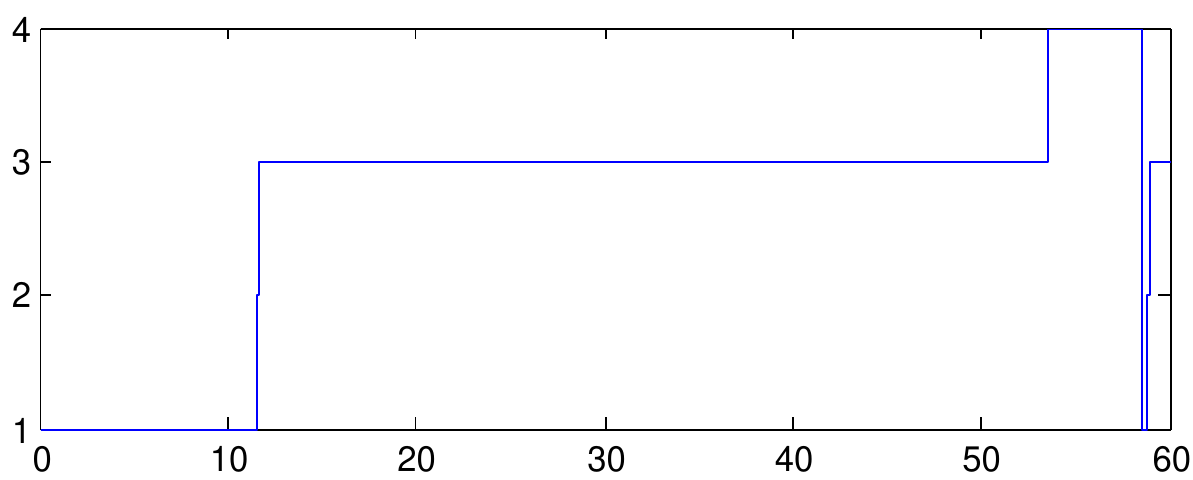}
\end{tabular}
\caption{Left: The shift scheduler of the powertrain challenge problem. 
Right: Shift schedules. The numbers correspond to the variables in the states of the shift scheduler.
Right Top: The shift schedule falsifying requirement $\phi^P_{e1}$.
Right Bottom: The shift schedule falsifying requirement $\phi^P_{e3}[0.4273]$.}
\label{fig:powertrain:graphs}
\vspace{-10pt}
\end{figure*}
 }
{}

\ifthenelse{\boolean{TECHREP}}{
In \cite{FainekosSUY12acc}, we demonstrated that \staliro{} \cite{AnnapureddyLFS11tacas,staliro:Online} can successfully solve the challenge problem (see Fig. \ref{fig:powertrain:graphs}) by formalizing the requirement as an MTL specification $\phi^P_{e1} = \neg \Diamond (g_2 \wedge \Diamond ( g_1 \wedge \Diamond g_2))$, where $g_i$ is a proposition that is true when the system is in gear $i$.
Stochastic search methods can be applied to solve the resulting optimization problem where the cost function is the robustness of the specification.}
{
In \cite{FainekosSUY12acc}, we demonstrated that \staliro{} can successfully solve the challenge problem.
}
Moreover, inspired by the success of \staliro on the challenge problem, we tried to ask a more complex question.
Specifically, does a transition exist from gear two to gear one and back to gear two in less than 2.5 sec?
\ifthenelse{\boolean{TECHREP}}{
An MTL specification that can capture this requirement is
$ \phi^P_{e2} = \Box ((\neg g_1 \wedge X g_1) \rightarrow  \Box_{(0,2.5]} \neg g_2)$.
The natural question that arises is what would be the smallest time for which such a transition can occur? 
}{}
We can formulate a parametric MTL formula to query the model of the powertrain system:
$ \phi^P_{e3}[\theta] = \Box ((\neg g_1 \wedge X g_1) \rightarrow  \Box_{(0,\theta]} \neg g_2) $.
\ifthenelse{\boolean{TECHREP}}{
We have extended \staliro to be able to handle parametric MTL specifications.
The total simulation time of the model is set to $60\sec$ and the search interval is $\Theta = [0,60]$. 
\staliro returned $\theta^* \approx 0.4273$ as the minimum parameter found \ifthenelse{\boolean{TECHREP}}{(See Fig. \ref{fig:powertrain:graphs})}{} using about 300 tests of the system.
}
{
Using \staliro we obtained $\theta^* \approx 0.4273$ as the minimum parameter found	using about 300 tests.
}

The challenge problem is extended to an industrial size high-fidelity engine model. The model is part of the SimuQuest Enginuity \cite{Simuquest:Online} Matlab/Simulink tool package. The Enginuity tool package includes a library of modules for engine component blocks. It also includes pre-assembled models for standard engine configurations\ifthenelse{\boolean{TECHREP}}{, see Fig. \ref{Fig:exmp:simuquestEnginuityEngine}}{}. In this work, we will use the Port Fuel Injected (PFI) spark ignition, 4 cylinder inline engine configuration. It models the effects of combustion from first physics principles on a cylinder-by-cylinder basis, while also including regression models for particularly complex physical phenomena. Simulink reports that this is a 56 state model. The model includes a tire-model, brake system model, and a drive train model (including final drive, torque converter and transmission). The model is based on a zero-dimensional modeling approach so that the model components can all be expressed in terms of ordinary differential equations. The inputs to the system are the throttle and brake schedules, and the road grade, which represents the incline of the road. The outputs are the vehicle and engine speed, the current gear and a timer that indicates the time spent on a gear. We search for a particular input for the throttle schedule, brake schedule, and grade level. The inputs are parametrized using 12 search variables, where 7 are used to model the throttle schedule, 3 for the brake schedule, and 2 for the grade level. The search variables for each input are interpolated with the Piecewise Cubic Hermite Interpolating Polynomial (PCHIP) function provided as a Matlab function by Mathworks. The simulation time is 60s. We demonstrate the parameter mining method for two specifications: 

\noindent 1. The specification $ \phi^{S}_{1}[\theta]=\Box_{[0,60]}((g2~\wedge~X g1) \rightarrow \Box_{[0,\theta]}((\tau \le \theta) \rightarrow g1)$, where $\tau$ is the time spent in a gear. The specification states that after shifting into gear one from gear two, there should be no shift from gear one to any other gear within $\theta$ seconds. 
\ifthenelse{\boolean{TECHREP}}{Clearly, the property defined is equivalent to the property defined in the challenge problem in the sense that the set of trajectories that satisfy/falsify the property is the same. The reason for the change made is the improved performance of the hybrid distance metric \cite{AbbasF11atva} with the modified specification.}{} The mined parameter for the specification returned is $1.29s$. 
Figure \ref{Fig:exmp:simuquest_param_est_gears} presents a shift schedule for which a transition out of gear one occurs in $1.28$ seconds.

\noindent 2. The specification $ \phi^{S}_{2}[\vec{\theta}]=\Box ((v < \theta_1)  \wedge (\omega < \theta_2))$, where $\theta_1$, $\theta_2$ represent the vehicle and engine speed parameters, respectively. The specification states that the vehicle and engine speed is always less than $\theta_1$ and $\theta_2$, respectively. The mined parameters for the specification returned are $137.1$mph and $4870$rpm.

\ifthenelse{\boolean{TECHREP}}{}{}
\begin{figure}
\begin{center}
\includegraphics[width=7.5cm]{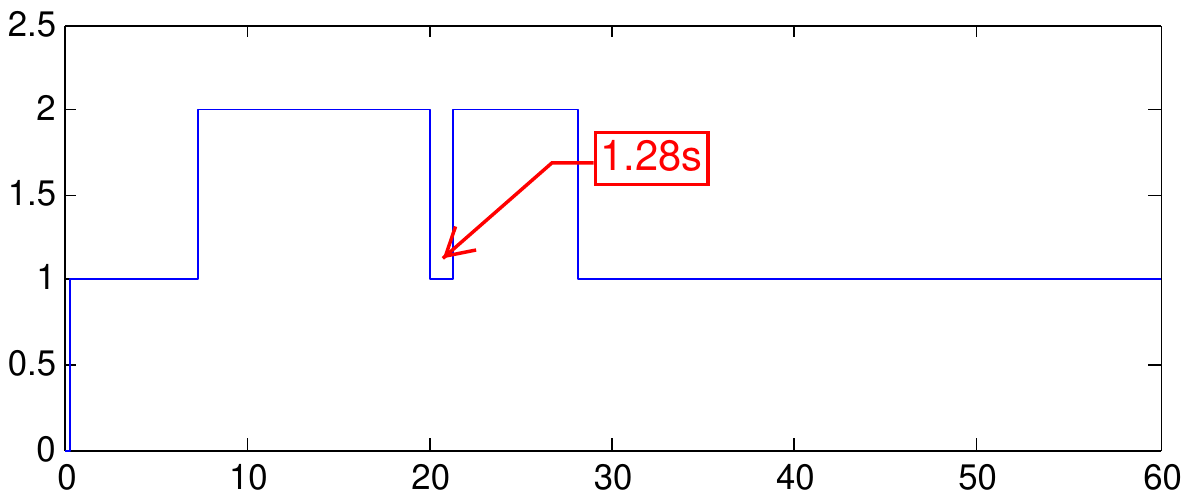} 
\end{center}
\caption{A shift schedule which falsifies the specification $\phi^{S}_{1}[\theta=1.29]=\Box_{[0,60]}((g_2~\wedge~X g_1) \rightarrow \Box_{[0,1.29]}((\tau \le 1.29) \rightarrow g_1)$ on the Simuquest high-fidelity engine model for specification.}
\label{Fig:exmp:simuquest_param_est_gears}
\vspace{-20pt} 
\end{figure}

In Table \ref{Tab:expts:staliro_breach}, we present experimental results for specifications on the Powertrain, Automotive Transmission, and Simuquest Enginuity high-fidelity engine models. 
\ifthenelse{\boolean{TECHREP}}{A detailed description of the benchmark problems can be found in \cite{AbbasFSIG11tecs,SankaranarayananF2012hscc} and the benchmarks can be downloaded with the \staliro distribution \cite{staliro:Online}.}

\ifthenelse{\boolean{TECHREP}}{
\begin{figure}
\begin{center}
\includegraphics[width=\columnwidth]{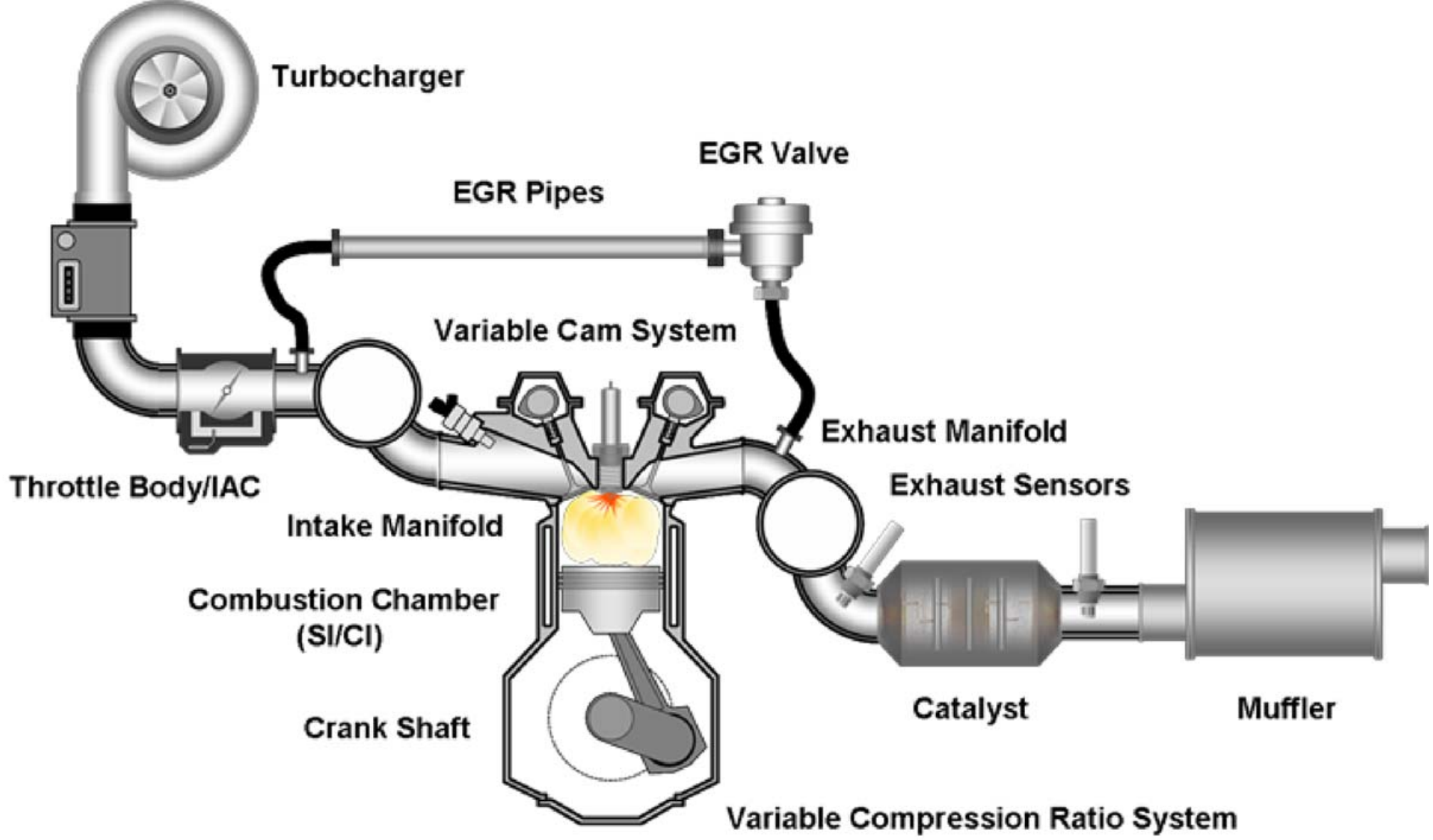} 
\end{center}
\caption{SimuQuest Enginuity model components. Used with permission, \copyright SimuQuest~\cite{Simuquest:Online}.}
\label{Fig:exmp:simuquestEnginuityEngine}
\vspace{-5pt} 
\end{figure}
}{}

\begin{table*}
\begin{center}
\caption{%
Experimental results of Parameter Mining with S-TaLiRo. The parameters were mined by running 1000 tests. Legend: $\bf{f(\vec{\theta}):}$ the priority function used, $\bf{\phi^{AT}_i:}$ Specifications tested on the Automotive Transmission Model, $\bf{\phi^{P}:}$ Specification tested on the Powertrain Model,  $\bf{\phi^{S}:}$ Specification tested on the Simuquest Enginuity high-fidelity Engine Model. The gray colored rows are first presented in \cite{YangHF12ictss} and are included for completeness.
}
\begin{tabular}{||c||c|c|c||}
\hline
\multicolumn{1}{||c||}{  } & \multicolumn{3}{c||}{S-TaLiRo} \\
\hline
 Specification & $f(\theta)$ & Time & Parameters Mined\\
\hline
\rowcolor{Gray}
$ \phi^{AT}_1[\theta] = \neg \Diamond  ( (v \ge 120) \wedge \Diamond_{[0,\theta]} (\omega \ge 4500) ) $ & $\theta$ & 135s & $7.7s$  \\
\hline
\rowcolor{Gray}
$ \phi^{AT}_2[\theta] = \neg \Diamond  ( (v \ge 120) \wedge \Diamond_{[0,\theta]} (v \ge 125)) $ & $\theta$ & 138s & $10.00s$    \\
\hline
\rowcolor{Gray}
$ \phi^{AT}_3[\theta] = \neg \Diamond  ( (v \ge 120) \wedge \Diamond_{[0,\theta]} (\omega \ge 4500)) $ & $\theta$ & 137s & $7.57s$    \\
\hline
\rowcolor{Gray}
$ \phi^{AT}_4[\theta] = \neg \Diamond  ( (v \ge 120) \wedge \Diamond_{[0,\theta]} (\omega \ge 4500)) $ & $\theta$ & 132s & $7.56s$    \\
\hline
\multirow{5}{*}{$\phi^{AT}_{5}[\vec{\theta}]=\Box ( ( v \le \theta_1) \wedge ( \omega \le \theta_2))$} & $\|\vec{\theta}\|$ & 139s & $\tupleof{138mph,5981rpm}$  \\
\cline{2-4}
	& $\theta_1$  & 137s & $\tupleof{57mph,6000rpm}$ \\
\cline{2-4}
	& $\theta_2$  & 138s & $\tupleof{180mph,2910rpm}$ \\
\cline{2-4}
	& $max(\vec{\theta})$ & 138s & $\tupleof{109mph,6000rpm}$ \\
\cline{2-4}
	& $min(\vec{\theta})$ & 138s & $\tupleof{154mph,5300rpm}$ \\
\hline
\multirow{5}{*}{$\phi^{AT}_{6}[\vec{\theta}] = \neg (\Diamond_{[0,\theta_1]}(v \ge 100) \wedge \Box (\omega \le \theta_2))$} & $\|\vec{\theta}\|$  & 144s & $\tupleof{15.7s,4820rpm}$ \\
\cline{2-4}
	& $\theta_1$  & 142s & $\tupleof{44.6s, 3598rpm}$  \\
\cline{2-4}
	& $\theta_2$  & 138s & $\tupleof{12.2s, 6000rpm}$  \\
\cline{2-4}
	& $max(\vec{\theta})$ & 140s & $\tupleof{37.3s,3742rpm}$  \\
\cline{2-4}
	& $min(\vec{\theta})$ & 142s & $\tupleof{12.3s,5677rpm}$ \\
\hline
\multirow{3}{*}{\parbox{7cm}{\centering $\phi^{AT}_{7}[\vec{\theta}] = \Box( (v \le \theta_1) \wedge (\omega \le \theta_2)) \wedge \Diamond_{[0,\theta_3]}(v \ge 150) \wedge \Diamond_{[0,\theta_4]}(\omega \ge 4500)$}} & $\|\vec{\theta}\|$ & 145s & $\tupleof{198mph,4932rpm,59.5s,55s}$ \\
\cline{2-4}
	& $max(\vec{\theta})$ & 143s & $\tupleof{129mph,6000rpm,48.9s,28.3s}$ \\
\cline{2-4}
	& $min(\vec{\theta})$ & 142s & $\tupleof{190mph,5575rpm,55.1s,54.8s}$ \\
\hline
\multirow{3}{*}{\parbox{7cm}{\centering $\phi^{AT}_{8}[\vec{\theta}] = \Box((v \le \theta_1) \wedge (\omega \le \theta_2)) \wedge \Diamond_{[0,\theta_3]}(v \ge 150) \wedge \Diamond_{[0,\theta_4]}(\omega \ge 4500) \wedge \Box_{[\theta_5,60]}(v \ge 170) \wedge \Box_{[\theta_6,60]}(\omega \ge 4750) $}} & $\|\vec{\theta}\|$ & 146s & $\tupleof{159mph,5700rpm,48.3s,36.2s,54.2s,53.9s}$ \\
\cline{2-4}
	& $max(\vec{\theta})$ & 145s & $\tupleof{85.9mph,6000rpm,3.8s,38.8s,44.5s,51.5s}$ \\
\cline{2-4}
	& $min(\vec{\theta})$ & 143s & $\tupleof{191mph,4958rpm,43s,55.3s,42s,47.1s}$ \\
\hline
$ \phi^{P}_{e3}[\theta] = \Box ((\neg g_1 \wedge X g_1) \rightarrow  \Box_{(0,\theta]} \neg g_2) $ & $\theta$ & 2600s & $0.1s$ \\
\hline
$ \phi^{S}_{1}[\theta]=\Box_{[0,60]}((g_2~\wedge~X g_1) \rightarrow \Box_{[0,\theta]}((t \le \theta) \rightarrow g_1)$ & $\theta$ & 21803s & $1.29s$ \\
\hline
\end{tabular}
\label{Tab:expts:staliro_breach}
\end{center}
\end{table*}

\begin{table*}
\begin{center}
\caption{%
Experimental Comparison of the method presented in this paper ($\mathcal{A}$) and the parameter synthesis method presented in \cite{jin2013mining}, ($\mathcal{B}$). Legend: \textbf{\#Sim.}: the number of system simulations, \textbf{\#Rob}: the number of robustness computations. 
}
\begin{tabular}{||c|c||c|c||c|c|c||}
\hline
\multicolumn{1}{||c|}{Specification} & \multicolumn{1}{|c||}{Method} & \multicolumn{2}{c||}{Parameters Mined} & {Time} & {\#Sim} & {\#Rob}  \\
\hline
\multirow{2}{*}{$ \phi^{S}_{2}[\vec{\theta}]=\Box ( ( v \le \theta_1) \wedge ( \omega \le\theta_2))$} &
$\mathcal{A}$ & 137.1 mph & 4870 rpm & 20170s  & 1000 & 1000  \\
\cline{2-7}
& $\mathcal{B}$ & 149.8 mph & 4883 rpm & 50017s  & 2386 & 5130  \\
\hline
\multirow{2}{*}{$\phi^{AT}_{5}[\vec{\theta}]=\Box ( ( v \le \theta_1) \wedge ( \omega \le\theta_2))$} &
$\mathcal{A}$ & 100.2 mph & 5987.6 rpm & 106s  & 1000 & 1000  \\
\cline{2-7}
& $\mathcal{B}$ & 137.5 mph & 6000 rpm & 253s  & 2176 & 11485  \\
\hline
\multirow{2}{*}{$\phi^{AT}_{6}[\vec{\theta}] = \neg (\Diamond_{[0,\theta_1]}(v \ge 100) \wedge \Box (\omega \le \theta_2)$} &
$\mathcal{A}$ & 21s & 3580 rpm & 110s  & 1000 & 1000  \\
\cline{2-7}
& $\mathcal{B}$ & 59.06s & 3296 rpm & 397s  & 3443 & 9718  \\
\hline
\end{tabular}
\label{Tab:relwork}
\end{center}
\end{table*}

%% file: related.tex

\section{Related Work}
%
%


The topic of testing embedded software and, in particular, embedded control software is a well studied problem that involves many subtopics well beyond the scope of this paper.
We refer the reader to specialized book chapters and textbooks for further information \cite{ConradF08crc,Koopman10book}.
Similarly, a lot of research has been invested on testing methods for Model Based Development (MBD) of embedded systems \cite{TripakisD09model}.
However, the temporal logic testing of embedded and hybrid systems has not received much attention \cite{TanKSL04,PlakuKV09tacas,NghiemSFIGP10hscc,ZulianiPC10hscc}.

Parametric temporal logics were first defined over traces of finite state machines \cite{Alur01tcl}.
In parametric temporal logics, some of the timing constraints of the temporal operators are replaced by parameters.
Then, the goal is to develop algorithms that will compute the values of the parameters that make the specification true under some optimality criteria.
That line of work has been extended to real-time systems and in particular to timed automata \cite{GiampaoloTN10lata} and continuous-time signals \cite{AsarinDMN12rv}.
The authors in \cite{Fages2008tcs,RizkBFS08cmsb} define a parametric temporal logic called quantifier free Linear Temporal Logic over real valued signals.
However, they focus on the problem of determining system parameters such that the system satisfies a given property rather than on the problem of exploring the properties of a given system.

Another related problem is specification mining or model exploration for finite state machines. 
The problem was initially introduced by William Chan in \cite{Chan00cav} under the term Temporal Logic Queries.
The goal of model exploration is to help the designer achieve a better understanding and explore the properties of a model of the system.
Namely, the user can pose a number of questions in temporal logic where the atomic propositions are replaced by a placeholder and the algorithm will try to find the set of atomic propositions for which the temporal logic formula evaluates to true.
Since the first paper \cite{Chan00cav}, several authors have studied the problem and proposed different versions and approaches \cite{BrunsG01lics,ChechikG03cav,GurfinkelDC02sigsoft,SinghRS09concur}.
A related approach is based on specification mining over temporal logic templates \cite{WasylkowskiZ09ase} rather than special placeholders in a specific formula. In \cite{kong2014temporal}, the authors present an inference algorithm that finds temporal logic properties of a system from data. The authors introduce a reactive parameter signal temporal logic and define a partial order over it to aid the property definition process.

In \cite{jin2013mining}, the authors provide a parameter synthesis algorithm for Parametric Signal Temporal Logic (PSTL), a similar formalism to MTL. To conduct parameter synthesis for multiple parameters, a binary search is utilized to set the parameter value for each parameter in sequence. After a set of parameters is proposed, a stochastic optimization algorithm is utilized to search for trajectories that falsify the specification. If it fails to do so, the algorithm stops, otherwise this two step process continues until the termination condition is met. 

In the following, we present three main differences between the method proposed here ($\mathcal{A}$) and the method proposed in \cite{jin2013mining} ($\mathcal{B}$). 
First, $\mathcal{A}$ is a best effort algorithm for which the termination condition is the number of tests the system engineer is interested to conduct. Clearly, the more tests, the better the search space is explored. Since the parameter mining problem is presented as a single optimization problem, runtime is not directly affected by the number of parameters in the specification. In contrast, in $\mathcal{B}$, the runtime of the algorithm through binary search is affected by the number of parameters in the PSTL formula. For each iteration of the binary search, multiple robustness computations have to be conducted, which for systems that output a large trace and contain complex specifications, could become costly. The second step in $\mathcal{B}$ is the falsification of the parameters proposed. This algorithm needs to be performed on every iteration, until a falsification is found. If a falsifying trajectory is not found, the stopping condition is met and the parameters are returned. Second, in $\mathcal{A}$, the parameters returned are the ``best" parameters for which a falsifying trajectory is found. In $\mathcal{B}$, the proposed parameters are parameters for which no falsifying trajectory is found. Proving that a specification holds for hybrid systems, in general, is undecidable and, therefore the failure to find a falsifying trajectory does not imply that one does not exist. Third, in $\mathcal{A}$, through the priority function, we enable the system engineer to have flexibility when assigning weights and priorities to parameters. In $\mathcal{B}$, parameter synthesis through binary search implicitly prioritizes one parameter over others. 

We compare the two methods using the Simuquest Enginuity high-fidelity Engine model and the Automotive Transmission model. To enable the comparison of the two methods, we have implemented the $\mathcal{B}$ method in S-TaLiRo. Note that the simulation time is 60s. The experimental results are presented in Table~\ref{Tab:relwork}. 
For the $\mathcal{A}$ method, the number of simulations and robustness computations is predefined. On the other hand, for the $\mathcal{B}$ method, these numbers vary following the reasons presented in the previous paragraph. As a result, the difference in computation time between the two methods is significant. Due to the significant differences between the two algorithms, in terms of guarantees provided, it is not possible to compare the quality of the solutions. While the mined parameters with method $\mathcal{A}$ guarantee falsification of the specification, the mined parameters with method $\mathcal{B}$ do not.

\ifthenelse{\boolean{TECHREP}}{
The results for the Automotive Transmission model can be reproduced by running the experiments in the S-TaLiRo distribution \cite{staliro:Online}. 
}{}

%% file: conclusions.tex
\section{Conclusion}

An important stage in Model Based Development (MBD) of software for CPS is the formalization of system requirements.
We advocate that Metric Temporal Logic (MTL) is an excellent candidate for formalizing interesting design requirements.
In this paper, we have presented a solution on how we can explore system properties using Parametric MTL (PMTL) \cite{AsarinDMN12rv}.
Based on the notion of robustness of MTL \cite{FainekosP09tcs}, we have converted the parameter mining problem into an optimization problem which we approximate using \staliro \cite{AnnapureddyLFS11tacas,staliro:Online}. 
We have presented a method for mining multiple parameters as long as the robustness function has the same monotonicity with respect to all the parameters.
Finally, we have demonstrated that our method can provide interesting insights to the powertrain challenge problem \cite{ChutinanB02fordtech}.We demonstrated the method on an industrial size engine model and examples from related works.